\documentclass[12pt]{article}

\usepackage[margin=2.5cm]{geometry}

\usepackage[T1]{fontenc}
\usepackage[utf8]{inputenc}
\usepackage{lmodern}
\usepackage{microtype}

\usepackage{amsmath, amssymb, amsthm, amsfonts}
\usepackage{bm}           
\usepackage{bbm}          
\usepackage{dsfont}        
\usepackage{nicefrac}      

\usepackage{graphicx}
\usepackage{subcaption}    
\usepackage{booktabs}      
\usepackage{longtable}
\usepackage{multirow}

\usepackage{algorithm}
\usepackage{algpseudocode}

\usepackage[dvipsnames]{xcolor}
\PassOptionsToPackage{hyphens}{url}
\usepackage[
  bookmarks=true,
  breaklinks=true,
  colorlinks=true,
  linkcolor=blue,
  citecolor=Blue,
  urlcolor=Blue
]{hyperref}
\usepackage{bookmark}

\usepackage[]{natbib}
\theoremstyle{plain}
\newtheorem{theorem}{Theorem}
\newtheorem{lemma}[theorem]{Lemma}

\newtheorem{corollary}[theorem]{Corollary}

\theoremstyle{definition}

\newtheorem{assumption}{Assumption}

\theoremstyle{remark}

\DeclareMathOperator*{\argmin}{arg\,min}

\newcommand{\R}{\mathbb{R}}

\newcommand{\dt}{\mathsf{d}}
\def\ss{\boldsymbol{s}}
\def\tt{\boldsymbol{\theta}}

\hypersetup{
  pdftitle={Preconditioned Robust Neural Posterior Estimation for Misspecified Simulators},
  pdfauthor={Ryan Kelly},
  pdfkeywords={approximate Bayesian computation, model misspecification, neural posterior estimation, random forests, simulation-based inference},
}

\begin{document}

\title{\bf Preconditioned Robust Neural Posterior Estimation for Misspecified Simulators}

\author{
  Ryan P. Kelly\thanks{Email: \texttt{kellyrp@qut.edu.au}.
  }\hspace{.2cm}\\
  School of Mathematical Sciences,
  Queensland University of Technology\\
  and \\
  David T. Frazier \\
  Department of Econometrics and Business Statistics,
  Monash University\\
  and \\
  David J. Warne \\
  School of Mathematical Sciences,
  Queensland University of Technology\\
  and \\
  Christopher C. Drovandi \\
  School of Mathematical Sciences,
  Queensland University of Technology
}

\maketitle

\begin{abstract}
Simulation-based inference (SBI) enables parameter estimation for complex stochastic models with intractable likelihoods when model simulation is feasible. Neural posterior estimation (NPE) is a popular SBI approach that often achieves accurate inference with far fewer simulations than classical approaches. But in practice, neural approaches can be unreliable for two reasons: incompatible data summaries arising from model misspecification yield unreliable posteriors due to extrapolation, and prior-predictive draws can produce extreme summaries that lead to difficulties in obtaining an accurate posterior for the observed data of interest. Existing preconditioning schemes target well-specified settings, and their behaviour under misspecification remains unexplored. We study preconditioning under misspecification and propose preconditioned robust neural posterior estimation, which computes data-dependent weights that focus training near the observed summaries and fits a robust neural posterior approximation. We also introduce a forest-proximity preconditioning approach that uses tree-based proximity scores to down-weight outlying simulations and concentrate computation around the observed dataset. Across two synthetic examples and one real example with incompatible summaries and extreme prior-predictive behaviour, we demonstrate that preconditioning combined with robust NPE increases stability and improves accuracy, calibration, and posterior-predictive fit over standard baseline methods.
\end{abstract}




\section{Introduction}\label{sec:intro}

Mechanistic simulators are essential tools for modelling natural phenomena across diverse scientific fields. While increased computing power has enabled the development of complex simulators that faithfully encode domain knowledge, this complexity often renders the explicit likelihood function intractable. Simulation-based inference (SBI) enables Bayesian inference in this setting by utilising model simulations to approximate the posterior distribution \citep{cranmer_frontier_2020}. Traditional SBI methods, such as approximate Bayesian computation (ABC), estimate the posterior by comparing observed and simulated data through low-dimensional summary statistics and a discrepancy metric \citep{sisson_handbook_2018}. However, these approaches are sensitive to tuning parameters and typically require a prohibitive number of simulations as the dimensionality of the summary statistics or model parameters increases \citep{blum_approximate_2010, barber_rate_2015}.

In contrast, neural approaches overcome these computational bottlenecks by learning flexible conditional densities from simulated data. These methods, most notably neural posterior estimation \citep[NPE;][]{greenberg_automatic_2019, papamakarios_fast_2016}, employ expressive conditional density estimators such as normalising flows to target the posterior directly. Unlike ABC, which rejects simulations that do not match the observation, NPE utilises all simulations to learn a density estimator over the entire data space. Once trained, evaluating the posterior at the observation is computationally efficient \citep{radev_bayesflow_2023}. This allows neural methods to deliver favourable accuracy-budget trade-offs, often achieving accurate inference with far fewer simulations than classical approaches \citep{lueckmann_benchmarking_2021}. Applications span diverse scientific domains, for instance, NPE is used in astrophysics to infer reionisation parameters from early-universe radio-signal simulations \citep{greig_exploring_2024}, in neuroscience to constrain generative models of cortical connectivity \citep{boelts_simulation-based_2023}, and in hydrology to calibrate process-based watershed simulators \citep{hull_simulation-based_2024}.

\citet{frazier_statistical_2024} demonstrate that NPE can achieve statistical accuracy comparable to exact Bayesian inference, provided the number of simulations scales sufficiently with the sample size, and provided the training data covers the observed summary statistics. However, when the observation lies in the extreme tail of the prior-predictive distribution, due to model misspecification, the density estimator is forced to extrapolate from the available training data. Broad priors present an additional challenge by generating global simulation spaces containing numerical pathologies that waste the neural density estimator's representational capacity. For example, in a detailed L5 pyramidal-cell model \citep{beck_efficient_2022}, 99.98\% of prior-predictive draws contained undefined features requiring imputation. Similarly, simulated residuals in pulsar-timing analyses spanned nearly $14$ orders of magnitude \citep{shih_fast_2024}, necessitating aggressive rescaling and clipping during training. In the presence of extreme simulated data, we might wish to constrain the training focus to the relevant summary space, effectively reintroducing the region-selection logic of ABC to assist the density estimator.

Beyond the computational burden of broad priors, structural model misspecification fundamentally alters the inferential goal. In this $\mathcal{M}$-open scenario, posterior updating concentrates on a pseudo-truth defined by the closest representable approximation \citep{berk_limiting_1966, walker_bayesian_2013}. Here, exact data reproduction is neither feasible nor required; rather, the inference must distinguish between relevant features that inform the target and irrelevant mismatches (e.g., noise or unmodelled artefacts). However, standard Bayesian updating lacks this discrimination, distorting the posterior to accommodate any type of discrepancy in mismodelled features. In the context of SBI, this structural incompatibility means that the observed summaries lie in the extreme tails or outside the model's support \citep{cannon_investigating_2022}. Reliable inference thus requires a strategy that mitigates the impact of irrelevant mismatches while concentrating around the true or pseudo-true parameter value based on relevant features.

To address the computational inefficiency and instability caused by broad priors, we adopt the preconditioning framework of \citet{wang_preconditioned_2024}. This strategy restricts neural training to the local neighbourhood of the observed summary statistics, avoiding the vast and potentially uninformative global prior-predictive distribution. The preconditioning is achieved by employing a relatively computationally inexpensive pilot run, such as a coarse SMC-ABC step, to filter out simulations far from the observation. Crucially, as ABC naturally concentrates on the pseudo-truth even under misspecification \citep{frazier_model_2020}, this preconditioning step identifies the parameter region yielding summaries closest to the observation, discarding irrelevant or pathological draws. By filtering the training data, preconditioning forces the density estimator to allocate its limited capacity primarily to the relevant regions of the parameter space.

While \citet{wang_preconditioned_2024} introduced this preconditioning framework for the well-specified setting, its behaviour under misspecification remains unexplored. Conversely, existing robust SBI methods typically operate on the global simulation space, leaving them susceptible to extreme prior-predictive draws. This highlights a gap for methods that jointly address misspecification and prior-predictive efficiency. We propose preconditioned robust neural posterior estimation (PRNPE), an approach that integrates preconditioning with robust inference. This method employs data-dependent weights to concentrate training resources on relevant simulations, while mitigating the impact of incompatible summaries through robust estimation.

In this work, we demonstrate the utility of preconditioning in misspecified settings, validating our approach on synthetic examples and a real-world application characterised by incompatible summaries. We also introduce forest-proximity preconditioning, an adaptive method that leverages tree structures to isolate the relevant simulation neighbourhood without needing to specify tolerances or distances as in the ABC preconditioning approach of \citet{wang_preconditioned_2024}. 
Furthermore, we formalise the theoretical justification for this summary-based weighting by bounding the amortisation gap in terms of weighted moments of the distance between training and observed summaries, and establishing the conditional invariance of the reweighted design.

In Section~\ref{sec:background}, we provide the necessary background on SBI. Section~\ref{sec:prnpe} develops the theory of preconditioning and details the PRNPE framework, introducing our forest-proximity and SMC-ABC implementations. Section~\ref{sec:examples} evaluates the method on two synthetic examples and a real-world application, and Section~\ref{sec:discussion} concludes with limitations and directions for future work.

\section{Simulation-based Bayesian inference}\label{sec:background}
Let $\bm{y} = (y_1, \ldots, y_n)^{\top} \in \mathcal{Y}$ be generated by an unknown data-generating process (DGP) $P_{\star}$ with density $p_{\star}$.
We posit a parametric family $\mathcal{P}=\{P_{\bm{\theta}} \colon \bm{\theta}\in\Theta\}$ with prior $\pi(\bm{\theta})$ and likelihood $p(\bm{y}\mid\bm{\theta})$, and target the posterior
 \begin{equation*}
     \pi(\bm{\theta} \mid \bm{y}) = Z(\bm{y})^{-1} p(\bm{y} \mid \bm{\theta}) \pi(\bm{\theta}), \quad Z(\bm{y}) = \int_{\Theta} p(\bm{y} \mid \bm{\theta}) \pi(\bm{\theta}) \, \mathrm{d}\bm{\theta}.
 \end{equation*}
For computational efficiency, we reduce the data dimension via a deterministic summary map $S:\mathcal{Y}\to\mathcal{S}\subset\mathbb{R}^{d_s}$. Let $\bm{s}_y=S(\bm{y})$ denote observed summaries and, for simulated data $\bm{x}\sim P_{\bm{\theta}}$, let $\bm{s}=S(\bm{x})$.
For $\bm{\theta}\in\Theta$, let $p(\bm{s}\mid\bm{\theta})$ denote the implied density of the summaries $S(\bm{X})$ when $\bm{X}\sim P_{\bm{\theta}}$, so that observed and simulated summaries share the summary space $\mathcal{S}$.
The joint density is $p(\bm{\theta},\bm{s})=\pi(\bm{\theta})\,p(\bm{s}\mid\bm{\theta})$ with marginal $p(\bm{s})=\int p(\bm{s}\mid\bm{\theta})\pi(\bm{\theta})\,\mathrm{d}\bm{\theta}$, the prior-predictive density of summaries.
Throughout, bold symbols denote vectors
and dependence on $n$ is suppressed.

\subsection{Approximate Bayesian computation}

ABC compares simulated and observed summaries via a discrepancy $\rho \colon \mathcal{S} \times \mathcal{S} \rightarrow [0, \infty)$, a kernel $K_{\epsilon}$, and a tolerance $\epsilon>0$ \citep{sisson_handbook_2018}.
The smoothed likelihood is
\begin{equation*}
 L_{\epsilon}(\bm{s}_y \mid \bm{\theta}) = \int K_{\epsilon}\big(\rho(\bm{s}, \bm{s}_y)\big)\, p(\bm{s} \mid \bm{\theta})\,\mathrm{d}\bm{s}.\end{equation*}
The ABC posterior is
\begin{equation*}
    \pi_{\epsilon}(\bm{\theta}\mid \bm{s}_y) \propto \pi(\bm{\theta})\,L_{\epsilon}(\bm{s}_y \mid \bm{\theta}).
    \end{equation*}
As $\epsilon \rightarrow 0$, the approximation converges to the partial posterior $\pi(\bm{\theta}\mid \bm{s}_y)$ \citep{barber_rate_2015}, which recovers the exact posterior $\pi(\bm{\theta}\mid \bm{y})$ if $S$ is sufficient \citep{blum_approximate_2010}.
However, for rejection ABC, the acceptance probability of a simulation from the prior-predictive scales as $O(\epsilon^{d_s})$, creating a curse-of-dimensionality trade-off between computational efficiency and accuracy \citep{barber_rate_2015, csillery_abc_2012}.
Naïve rejection becomes computationally infeasible for achieving an accurate approximation of the partial posterior when the prior is diffuse relative to the posterior.
This motivates Markov chain Monte Carlo (MCMC) ABC \citep{marjoram_markov_2003}, which improves proposal efficiency via local steps, or sequential Monte Carlo (SMC) ABC \citep{sisson_sequential_2007, beaumont_adaptive_2009}, which propagate particles through a sequence of decreasing tolerances. Although these more advanced implementations of ABC are more efficient than rejection ABC, they ultimately suffer from the same curse of dimensionality. However, for the preconditioning detailed in Section~\ref{sec:prnpe}, we require only a coarse filtering of the prior-predictive distribution and need not drive $\epsilon \rightarrow 0$, as a moderate tolerance suffices to down-weight regions with negligible support near $\bm{s}_y$. In this paper, we employ the replenishment SMC ABC algorithm of \citet{drovandi_estimation_2011}, detailed in Appendix~\ref{sec:app_smcabc_background}.

A popular class of ABC methods, ABC random forests (ABC-RF), replaces explicit distances and tolerances with random forests trained to predict parameters from summaries \citep{raynal_abc_2019}. We provide background on random forests in Appendix~\ref{sec:app_rf_background}.
These approaches are robust to the dimensionality and choice of summary statistics.
We train independent regression forests for each parameter coordinate $\theta_j$. These forests define data-adaptive weights for a simulation $\bm{s}_i$ based on its leaf co-occurrence with the observation $\bm{s}_y$. Averaging across $B$ trees and all $d_{\bm{\theta}}$ parameter forests, the weight is
\begin{equation*}
  W^{\mathrm{RF}}(\bm{s}_i;\bm{s}_y)=\frac{1}{d_{\bm{\theta}}B}\sum_{j=1}^{d_{\bm{\theta}}}\sum_{b=1}^{B}\frac{\mathbb{I}\{\bm{s}_i\in L_{jb}(\bm{s}_y)\}}{|L_{jb}(\bm{s}_y)|}, \quad \text{with } \tilde w_i\propto W^{\mathrm{RF}}(\bm{s}_i;\bm{s}_y),
\end{equation*}
where $L_{jb}(\bm{s}_y)$ denotes the leaf containing $\bm{s}_y$ in tree $b$ of the forest for $\theta_j$, and $|L|$ is the count of training samples in that leaf.
These weights define tolerance‑free, data‑adaptive neighbourhoods around $\bm{s}_y$ \citep{raynal_abc_2019,meinshausen_quantile_2006}.
While this presented weighting scheme will be central to our preconditioning, similar ABC-RF frameworks also extend to model choice \citep{pudlo_reliable_2016}, posterior regression adjustment \citep{bi_random_2022}, and distributional forests for joint posteriors \citep{dinh_approximate_2025}.

\subsection{Neural density estimation}

While standard ABC methods discard simulations that are not close to the observation, NPE utilises all available simulations to learn a global approximation $q_{\bm{\phi}}(\bm{\theta} \mid \bm{s}) \approx \pi(\bm{\theta} \mid \bm{s})$.
Given a dataset of joint draws $\{(\bm{\theta}_i, \bm{s}_i)\}_{i=1}^N \sim p(\bm{\theta}, \bm{s})$, we optimise the neural network parameters $\bm{\phi} \in \Phi$ by minimising the expected Kullback-Leibler divergence between the true and approximate posteriors. This is equivalent to maximising the log-likelihood of the parameters given the summaries:
\begin{equation}\label{eq:npe_loss_prnpe}
  \bm{\phi}^* = \argmin_{\bm{\phi} \in \Phi} \mathbb{E}_{(\bm{\theta}, \bm{s}) \sim p(\bm{\theta}, \bm{s})} \bigl[- \log q_{\bm{\phi}} (\bm{\theta} \mid \bm{s})\bigr].\end{equation}
Normalising flows \citep{rezende_variational_2015} are a common choice for $q_{\bm{\phi}}$ as they provide tractable density evaluation and efficient sampling.
Sequential variants (SNPE) iteratively refine proposals to focus on the observation \citep{papamakarios_fast_2016, lueckmann_flexible_2017, greenberg_automatic_2019}, a strategy recently extended to round-free, online updates via dynamic SBI \citep{lyu_dynamic_2025}. 
However, we focus here on NPE trained on the prior, which serves as the baseline for our preconditioning approach.
Recent theoretical results establish the consistency and asymptotic normality of NPE under standard regularity conditions \citep{frazier_statistical_2024}.

From an amortised inference perspective, Eq.~\eqref{eq:npe_loss_prnpe} corresponds to minimising the population risk under a specific simulation design.
A design specifies a parameter proposal $\bm{\theta} \sim r(\bm{\theta})$, which induces a joint training distribution $p_{\text{train}}(\bm{\theta}, \bm{s})$.
Standard NPE sets $r=\pi$, thereby averaging performance over the prior-predictive distribution.
However, accurate inference is strictly required only at the observed $\bm{s}_y$; the discrepancy between the global average loss and the local accuracy at $\bm{s}_y$ is termed the amortisation gap \citep{cremer_inference_2018, zammit-mangion_neural_2025}.
Consequently, the choice of design $r(\bm{\theta})$ governs which regions of the parameter space dominate the loss, determining where the estimator is most reliable.

An alternative approach, neural likelihood estimation (NLE), trains a conditional density estimator $q_{\bm{\phi}}(\bm{s}\mid\bm{\theta}) \approx p(\bm{s}\mid\bm{\theta})$. Unlike NPE, the learnt likelihood is reusable across different priors without retraining \citep{cranmer_approximating_2016, papamakarios_sequential_2019}.
For inference, we combine the surrogate likelihood with the prior,
\begin{equation*}
    \hat{\pi}(\bm{\theta} \mid \bm{s}_y) \propto q_{\bm{\phi}^*}(\bm{s}_y \mid \bm{\theta}). \end{equation*}
Posterior samples are then drawn using MCMC, substituting $q_{\bm{\phi}^*}(\bm{s}_y \mid \bm{\theta})$ for the intractable likelihood.
Sequential NLE (SNL) adapts the simulation proposals to concentrate training data in regions where the likelihood is high \citep{papamakarios_sequential_2019}.
Finally, neural ratio estimation (NRE) learns the likelihood-to-evidence ratio $r_{\bm{\phi}}(\bm{s}, \bm{\theta}) \propto p(\bm{s} \mid \bm{\theta}) / p(\bm{s})$ via binary classification between joint and product‑of‑marginals samples, enabling inference via MCMC or importance weighting \citep{hermans_likelihood-free_2020, durkan_contrastive_2020}.

\subsection{Robust inference under model misspecification}

Model misspecification arises when the true DGP, $P_{\star}$, lies outside the posited family $\mathcal{P}$.
Adapting standard definitions to the summary space \citep{marin_relevant_2014, frazier_model_2020}, we define the binding functions for the model, $b(\bm{\theta}) = \mathbb{E}_{P_{\bm{\theta}}}[S(\bm{x})]$, and the observation, $b_{\star} = \mathbb{E}_{P_{\star}}[S(\bm{y})]$.
We characterise the model as incompatible if the summary-level mismatch
 \begin{equation}\label{eq:incompatibility}
    \epsilon^{*} = \inf_{\bm{\theta} \in \Theta} \rho(b_{\star}, b(\bm{\theta}))
 \end{equation}
is strictly positive. While ABC concentrates on the pseudo-true parameter that minimizes this discrepancy \citep{frazier_model_2020}, incompatibility poses a distinct challenge for neural estimators (e.g., NPE/NLE).
These networks are trained on simulations from a global distribution $p_{\text{train}}(\bm{s})$.
If $\epsilon^{*} > 0$, the observed $\bm{s}_y$ likely falls in a region of $\mathcal{S}$ with low (or zero) support under $p_{\text{train}}$.
Inference at $\bm{s}_y$ therefore requires neural network extrapolation, which often yields overconfident or misleading posteriors \citep{cannon_investigating_2022, schmitt_detecting_2024}; this exacerbates the general failure of Bayesian credible sets to maintain frequentist coverage under misspecification \citep{kleijn_bernstein-von-mises_2012}.

Strategies for robust SBI include constructing robust summary statistics, minimising generalised Bayesian losses, or explicitly modelling the error between simulated and observed data \citep{kelly_simulation-based_2025}; we adopt the latter approach.
At the summary level, this is often formulated via an additive error term $\bm{\Gamma}$, such that $\bm{s}_y = \bm{s} + \bm{\Gamma}$.
This transforms the likelihood into a convolution of the model pushforward density $p(\bm{s} \mid \bm{\theta})$ and the error density $h(\bm{\Gamma})$,
\begin{equation*}
    p(\bm{s}_y \mid \bm{\theta}) = \int p(\bm{s} \mid \bm{\theta}) h(\bm{s}_y - \bm{s}) \, \mathrm{d}\bm{s}.
\end{equation*}
This formulation extends concepts from computer model calibration \citep{kennedy_bayesian_2001, bayarri_modularization_2009} and relates to the view of ABC as exact inference under a specific error model \citep{ratmann_model_2009, wilkinson_approximate_2013}.

Robust Bayesian synthetic likelihood (RBSL) modifies the standard Bayesian synthetic likelihood (BSL) framework to accommodate model misspecification \citep{frazier_robust_2021}.
While standard BSL approximates the summary likelihood as a Gaussian $p(\bm{s}_y \mid \bm{\theta}) \approx \mathcal{N}(\bm{s}_y \mid \bm{\mu}_{\bm{\theta}}, \bm{\Sigma}_{\bm{\theta}})$ \citep{price_bayesian_2018, wood_statistical_2010}, RBSL targets the joint posterior $\pi(\bm{\theta}, \bm{\Gamma} \mid \bm{s}_y) \propto p_{\text{RBSL}}(\bm{s}_y \mid \bm{\theta}, \bm{\Gamma}) \,\pi(\bm{\theta}) \pi(\bm{\Gamma})$, introducing parameters $\bm{\Gamma}$ to absorb the error.
In the mean-adjustment variant (RBSL-M), the likelihood mean is shifted to $\bm{\mu}_{\bm{\theta}} + \bm{\sigma}_{\bm{\theta}} \odot \bm{\Gamma}$, where $\odot$ denotes the element-wise product and $\bm{\sigma}_{\bm{\theta}}$ is the vector of summary standard deviations.
Placing a sparsity-promoting prior on $\bm{\Gamma}$, such as a Laplace distribution, shrinks adjustments towards zero.
Consequently, posterior components of $\bm{\Gamma}$ that deviate significantly from the prior centred at zero indicates misspecification, providing a diagnostic for incompatible summaries.

Robust neural posterior estimation (RNPE) extends NPE by specifying an explicit error model on the summaries and marginalising out the error through the pre-trained posterior \citep{ward_robust_2022}.
In addition to the conditional density $q_{\bm{\phi}}(\bm{\theta}\mid \bm{s})$, RNPE learns a density estimator $h_{\bm{\psi}}(\bm{s})$ for the marginal distribution of simulated summaries.
Defining the error $\bm{\delta} = \bm{s}_y - \bm{s}$, RNPE assumes a component-wise spike-and-slab prior with indicators $z_j \in \{0,1\}$,
\begin{equation*}    
\delta_j \mid z_j \sim (1-z_j)\,\mathcal{N}(0,\tau_j^2) + z_j\,\mathrm{Cauchy}(0,\kappa_j), \quad z_j \sim \mathrm{Bernoulli}(\pi_j).
\end{equation*}
Here, the posterior probability $p(z_j=1 \mid \bm{s}_y, \bm{s})$ identifies whether summary $j$ is misspecified. Inference proceeds by sampling ``denoised'' summaries $\tilde{\bm{s}}$ from the target
\begin{equation*}
\widehat{p}(\bm{s} \mid \bm{s}_y) \propto p(\bm{s}_y \mid \bm{s})\, h_{\bm{\psi}}(\bm{s}),
\end{equation*}
where $p(\bm{s}_y \mid \bm{s})$ is induced by the error model, and then generating parameters via $\tilde{\bm{\theta}} \sim q_{\bm{\phi}}(\bm{\theta} \mid \tilde{\bm{s}})$.
A closely related method, robust neural likelihood estimation (RNLE), similarly introduces additive adjustment parameters but targets the joint posterior of $(\bm{\theta}, \bm{\Gamma})$ using a learnt surrogate likelihood \citep{kelly_misspecification-robust_2024}.

\section{Methods}\label{sec:prnpe}
We propose precondition robust neural posterior estimation (PRNPE), where we precondition simulations around $S(\bm{y})$, then fit $q_{\bm{\phi}}(\bm{\theta}\mid\bm{s})$ and denoise $S(\bm{y})$ via RNPE.

\subsection{Problem setup and preconditioning}\label{sec:prnpe_setup}

We target single-instance inference at the observed summary $\bm{s}_y=S(\bm{y})$, seeking the posterior $\pi(\bm{\theta} \mid \bm{s}_y)$, under the assumed model components $(\pi(\cdot), p(\bm{x} \mid \bm{\theta}), S(\cdot))$.
To approximate this posterior for $\bm{s}_y$, we learn an amortised estimator $q_{\bm{\phi}}(\bm{\theta} \mid \bm{s})$.
Standard NPE minimises the population risk over a simulation design with parameter law $r(\bm{\theta})$:
\[
\mathcal L(\bm{\phi})=\mathbb E_{(\bm{\theta},\bm{s})\sim r(\bm{\theta})p(\bm{s}\mid\bm{\theta})}\left[-\log q_{\bm{\phi}}(\bm{\theta}\mid \bm{s})\right].
\]

Practitioners routinely discard numerically failed simulations (NaNs or $\pm\infty$) using simple checks or a validity classifier \citep{lueckmann_flexible_2017}, and we assume such failures are filtered out prior to training. A subtler failure mode remains when weakly informative priors or stochastic simulators generate extreme but finite summaries: these few large‑magnitude \(\bm{s}\) values dominate the \(-\log q_{\bm{\phi}}(\bm{\theta}\mid\bm{s})\) objective, destabilise gradients, and waste model capacity on regions of $\mathcal{S}$ irrelevant to the observation $\bm{s}_y$. A common pragmatic approach is to clip or discard extreme simulations \citep[e.g.,][]{de_santi_field-level_2025,shih_fast_2024}, but threshold choice is problem‑specific. We instead adopt a principled alternative, preconditioning on summaries as in \citet{wang_preconditioned_2024}, which concentrates simulation effort near \(\bm{s}_y\).

We define a data‑dependent weight function
\[
w_y \colon \mathcal{S} \rightarrow [0,\infty),
\]
parameterised by the observation; we write \(w_y(\bm{s})\) or \(w_y(\bm{s};\bm{s}_y)\) when we want to emphasise the dependence on \(\bm{s}_y\).
To balance the benefits of retaining already simulated training data against the cost of modelling irrelevant regions, we design \(w_y\) to be conservative, excising pathological extremes while preserving sufficient breadth to ensure robust generalisation.
We write
\[
p_{\text{train}}(\bm{\theta},\bm{s}) := r(\bm{\theta})p(\bm{s}\mid\bm{\theta}),\qquad
p_{\text{train}}(\bm{s}) := \int p_{\text{train}}(\bm{\theta},\bm{s})\,d\bm{\theta},
\]
for the baseline training design and define the reweighted design
\[
p_w(\bm{\theta},\bm{s}) \;\propto\; w_y(\bm{s})\,p_{\text{train}}(\bm{\theta},\bm{s}).
\]
We approximate this weighted risk via importance resampling.
Given a batch of \(N\) simulations \(\mathcal{D} = \{(\bm{\theta}_i, \bm{s}_i)\}_{i=1}^N\), we compute normalised weights
\(\tilde{w}_i \propto w_y(\bm{s}_i)\) and draw indices
\(I_1, \ldots, I_K \sim \mathrm{Categorical}(\tilde{\bm{w}})\) to form a resampled dataset.
We then minimise the empirical risk on this resample.

In \citet{wang_preconditioned_2024}, similar weighting is induced by a short SMC‑ABC run that accepts simulations whose discrepancy to \(\bm{s}_y\) falls below an adaptive threshold, and an unconditional density $q_{\bm{\phi}}(\bm{\theta})$ is fit to the ABC pilot as a proposal for sequential rounds. We extend this approach by training the conditional estimator $q_{\bm{\phi}}(\bm{\theta}\mid\bm{s})$ directly on the preconditioned resample, yielding a conditional posterior approximation without the intermediate unconditional fitting step.
Properties and consequences are developed in Section~\ref{sec:prnpe_theory}; concrete choices for obtaining \(w_y\) (forest-proximity and SMC‑ABC preconditioners) are given in Section~\ref{sec:prnpe_pipeline}.

\subsection{Theoretical properties of preconditioned NPE}\label{sec:prnpe_theory}
We assess the single-instance accuracy of an amortised estimator $q_{\bm{\phi}}(\bm{\theta}\mid\bm{s})$ at the observation $\bm{s}_y$.
Standard NPE obtains parameters $\hat{\bm{\phi}}$ by minimising the global population risk over a training design $p_{\text{train}}(\bm{\theta}, \bm{s})$:
\begin{equation}
\mathcal{L}(\bm{\phi}) := \mathbb{E}_{(\bm{\theta}, \bm{s}) \sim p_{\text{train}}} \big[ -\log q_{\bm{\phi}}(\bm{\theta} \mid \bm{s}) \big].
\end{equation}
However, our inferential goal is to minimise the local risk at the specific observation $\bm{s}_y$:
\begin{equation}
\mathcal{L}_{\bm{s}_y}(\bm{\phi}) := \mathbb{E}_{\bm{\theta} \sim \pi(\cdot\mid\bm{s}_y)} \big[ -\log q_{\bm{\phi}}(\bm{\theta} \mid \bm{s}_y) \big].
\end{equation}
Let
\[
A(\bm{s}_y) := \inf_{\bm{\phi}} \mathcal{L}_{\bm{s}_y}(\bm{\phi})
\]
denote the approximation floor (the best achievable single-instance log-loss within the variational family).
If $q_{\bm{\phi}}(\bm{\theta}\mid\bm{s}_y)$ can represent the true posterior exactly, then $A(\bm{s}_y)$ coincides with the entropy of $\pi(\bm{\theta}\mid\bm{s}_y)$; otherwise it includes an irreducible approximation bias.
We identify the amortisation gap \citep{cremer_inference_2018, zammit-mangion_neural_2025} as the excess risk incurred by minimising the global, rather than local, objective:
\begin{equation}
\Delta_{\mathrm{am}}(\bm{s}_y) := \mathcal{L}_{\bm{s}_y}(\hat{\bm{\phi}}) - A(\bm{s}_y).
\end{equation}
Since $\hat{\bm{\phi}}$ is tuned to minimise average loss under $p_{\text{train}}$, the magnitude of $\Delta_{\mathrm{am}}(\bm{s}_y)$ is driven by where $p_{\text{train}}(\bm{s})$ places mass in summary space. Broad priors that allocate many simulations to extreme prior-predictive summaries can enlarge the gap at $\bm{s}_y$ by diverting capacity away from its neighbourhood, even when $\bm{s}_y$ itself lies within the support.
For interpretation it is convenient to write
\begin{equation}
\Delta_{\mathrm{am}}(\bm{s}_y) = \Delta_{\mathrm{est}}(\bm{s}_y) + \Delta_{\mathrm{oos}}(\bm{s}_y),
\end{equation}
where $\Delta_{\mathrm{est}}(\bm{s}_y)$ captures estimation error governed by the training density, and $\Delta_{\mathrm{oos}}(\bm{s}_y) \ge 0$ represents extrapolation error caused by incompatible summaries lying outside the simulator's support.
This distinction mirrors the compatible and incompatible regimes in \citet{frazier_statistical_2024}, where only the former admits consistent posterior inference.
Preconditioning aims to reduce $\Delta_{\mathrm{est}}(\bm{s}_y)$ by concentrating training mass near the observed summary $\bm{s}_y$, but it cannot resolve incompatibility. Consequently, robust inference methods are required to mitigate $\Delta_{\mathrm{oos}}(\bm{s}_y) \ge 0$, often by mitigating the incompatible summaries and conducting inference on the compatible summaries (see Section~\ref{sec:prnpe_pipeline}), effectively bypassing the extrapolation penalty.

We now formalise the benefit of preconditioning.
Let $w_y\colon\mathcal{S}\to[0,\infty)$ be a data-dependent weight function and let $\bm{\phi}^{\star}_w$ denote the population minimiser of the weighted risk $\int \mathcal{L}(\bm{\phi};\bm{s})\,w_y(\bm{s})\,p_{\text{train}}(\bm{s})\,\mathrm{d}\bm{s}$; setting $w_y\equiv 1$ recovers the standard unweighted NPE minimiser.
The amortisation gap at $\bm{s}_y$ is then $\Delta_{\mathrm{am}}(\bm{s}_y) = \mathcal{L}_{\bm{s}_y}(\bm{\phi}^{\star}_w) - A(\bm{s}_y)$.
Under the regularity conditions detailed in Assumption~1 in Appendix~\ref{app:am_gap_proof}, which require a $\bm{\theta}$-dependent Lipschitz condition on $\bm{s}\mapsto\log q_{\bm{\phi}}(\bm{\theta}\mid\bm{s})$ and a H\"{o}lder smoothness condition on $\bm{s}\mapsto\pi(\bm{\theta}\mid\bm{s})$ with exponent $\kappa>1$, we obtain the following.
\begin{lemma}[Amortisation-gap bound]\label{lem:am_gap}
Under Assumption~1 in Appendix~\ref{app:am_gap_proof},
\begin{align}
    \Delta_{\mathrm{am}}(\bm{s}_y) \;\le\; 
    4\bar{C}_1\,\mathbb{E}_{p_{\text{train}}(\bm{s})}\big[w_y(\bm{s})\|\bm{s}-\bm{s}_y\|\big] 
    \;+\;& 2\bar{C}_2\sqrt{\mathbb{E}_{p_{\text{train}}(\bm{s})}\big[w_y(\bm{s})\|\bm{s}-\bm{s}_y\|^2\big]}\nonumber \\
    \;+\;& 2\bar{C}_3\,\mathbb{E}_{p_{\text{train}}(\bm{s})}\big[w_y(\bm{s})\|\bm{s}-\bm{s}_y\|^{\kappa}\big].
    \label{eq:gap_bound}
\end{align}
\end{lemma}
In the compatible regime where $\bm{s}_y$ lies in the effective support of $p_{\text{train}}(\bm{s})$, $\Delta_{\mathrm{oos}}(\bm{s}_y)=0$ and the bound applies directly to $\Delta_{\mathrm{am}}(\bm{s}_y)$. In general we can view \eqref{eq:gap_bound} as a bound on the on-support component $\Delta_{\mathrm{est}}(\bm{s}_y)$, with any extrapolation error $\Delta_{\mathrm{oos}}(\bm{s}_y)\ge 0$ left uncontrolled.
Here $\bar{C}_1$ controls the sensitivity of the pointwise loss to perturbations in $\bm{s}$ (via a $\bm{\theta}$-dependent Lipschitz constant integrated against $\pi(\bm{\theta}\mid\bm{s}_y)$), $\bar{C}_2$ captures a second-moment contribution through the weighted joint distribution $p_{\text{train}}(\bm{\theta},\bm{s})$, and $\bar{C}_3$ quantifies the regularity of the posterior density $\pi(\bm{\theta}\mid\bm{s})$ as a function of $\bm{s}$.
The three terms control the weighted first moment, weighted second moment, and weighted $\kappa$th moment of the distance $\|\bm{s}-\bm{s}_y\|$ under $p_{\text{train}}$.
All three are simultaneously reduced by concentrating $w_y$ near $\bm{s}_y$.
Indeed, if $w_y(\bm{s})\le C'\,\mathbb{I}\{\|\bm{s}-\bm{s}_y\|\le\epsilon\}$, the bound simplifies to $O(\epsilon\vee\epsilon^{\kappa})$ (Corollary~1 in Appendix~\ref{app:am_gap_proof}).

The first term generalises the transport cost of the design: when $w_y\equiv 1$,
\[
\mathbb{E}_{p_{\text{train}}(\bm{s})}\big[\|\bm{s}-\bm{s}_y\|\big]
= W_1\big(p_{\text{train}}(\bm{s}),\delta_{\bm{s}_y}\big)
\]
under the ground metric induced by $\|\cdot\|$.
When $p_{\text{train}}(\bm{s})$ is diffuse (for example, under a broad prior or a heavy-tailed simulator), this Wasserstein distance is large and gradients from distant, irrelevant simulations can dominate the loss.
The additional second-moment and H\"{o}lder terms in~\eqref{eq:gap_bound} reflect finer structure: the $\bar{C}_2$ term arises from variation in the conditional $\pi(\bm{\theta}\mid\bm{s})$ across training summaries (entering through the joint distribution rather than the marginal alone), while the $\bar{C}_3$ term penalises irregularity of the posterior as a function of the conditioning variable.
We view \eqref{eq:gap_bound} as a design-level guideline: the weighted moments of $\|\bm{s}-\bm{s}_y\|$ are explicit, geometry-driven contributions to the on-support component $\Delta_{\mathrm{est}}(\bm{s}_y)$ that can be reduced by concentrating training mass near $\bm{s}_y$.
We note that the constants $\bar{C}_1$, $\bar{C}_2$, $\bar{C}_3$ depend implicitly on the data size $n$ (through the posterior $\pi(\bm{\theta}\mid\bm{s})$), and are therefore best interpreted at a fixed sample size.
In incompatible regimes where $\bm{s}_y$ lies far outside the support of $p_{\text{train}}(\bm{s})$, the bound becomes uninformative.

This geometric insight directly motivates preconditioning.
By replacing the unweighted design with a weighted risk using $w_y(\bm{s})$ concentrated near $\bm{s}_y$, we shrink all three terms in~\eqref{eq:gap_bound} and improve the conditioning of the learning problem.
We note that $\bm{\phi}^{\star}_w$ minimises the weighted loss $\mathbb{E}_{(\bm{\theta},\bm{s})\sim p_{\text{train}}}[-w_y(\bm{s})\log q_{\bm{\phi}}(\bm{\theta}\mid\bm{s})]$, where the expectation is taken under the original training design $p_{\text{train}}$; importance resampling provides a finite-sample approximation to this weighted objective.
However, preconditioning cannot create support where $p(\bm{s}\mid\bm{\theta})=0$. If $\bm{s}_y\notin\operatorname{supp}(p_{\text{train}}(\bm{s}))$, the analysis remains extrapolative and the bound is uninformative, necessitating the robust inference stage described in Section~\ref{sec:prnpe_pipeline}.

An important requirement for the preconditioning schemes we develop is that they concentrate training mass near $\bm{s}_y$ without biasing the conditional relationship between parameters and summaries. We say that a weighting scheme is conditionally invariant if the reweighted design
\(p_w(\bm{\theta},\bm{s})\) satisfies
\[
p_w(\bm{\theta}\mid\bm{s}) = p_{\text{train}}(\bm{\theta}\mid\bm{s})
\quad\text{for all }\bm{s}\text{ with }w_y(\bm{s})>0.
\]
This holds whenever the weight depends only on summaries. Let
\(w_y\colon\mathcal{S}\to[0,\infty)\) and define the reweighted joint
\[
p_w(\bm{\theta},\bm{s}) \;\propto\; w_y(\bm{s})\,p_{\text{train}}(\bm{\theta},\bm{s}).
\]
For any \(\bm{s}\) with \(w_y(\bm{s})>0\),
\[
p_w(\bm{\theta}\mid\bm{s})
=\frac{w_y(\bm{s})\,p_{\text{train}}(\bm{\theta},\bm{s})}{\int w_y(\bm{s})\,p_{\text{train}}(\bm{\theta}',\bm{s})\,d\bm{\theta}'}
=\frac{w_y(\bm{s})\,p_{\text{train}}(\bm{\theta},\bm{s})}{w_y(\bm{s})\,p_{\text{train}}(\bm{s})}
=p_{\text{train}}(\bm{\theta}\mid\bm{s}).
\]
If \(w_y\) is indicator-valued (i.e.\ \(w_y(\bm{s})\in\{0,1\}\) for all \(\bm{s}\)) we simply filter simulations: the conditional \(p_w(\bm{\theta}\mid\bm{s})\) is unchanged on kept summaries with \(w_y(\bm{s})=1\) and undefined on discarded summaries with \(w_y(\bm{s})=0\). Since \(w_y(\bm{s}_y)>0\) by construction, the conditional target at \(\bm{s}_y\) is preserved.


Because the weights $w_y(\bm{s})$ act only in summary space, the amortisation-gap bound in~\eqref{eq:gap_bound} applies directly to the preconditioned design.
Concentrating $w_y$ near $\bm{s}_y$ tightens the weighted moments and typically reduces loss variance and gradient dispersion in the neighbourhood that matters.

Many practical preconditioning schemes fit this summary-only template. Random forest (RF) preconditioning defines weights
\(w_y(\bm{s}) = W^{\mathrm{RF}}(\bm{s};\bm{s}_y)\) via proximity in trees built on \(\bm{s}\); leaf assignments depend only on $\bm{s}$, so RF reweighting multiplies \(p_{\text{train}}(\bm{\theta},\bm{s})\) by a function of \(\bm{s}\) alone and leaves \(p_{\text{train}}(\bm{\theta}\mid\bm{s})\) invariant. Replenishment SMC-ABC generates particles targeting a sequence of joint distributions
\[
\pi_t(\bm{\theta},\bm{s}) \;\propto\; \pi(\bm{\theta})p(\bm{s}\mid\bm{\theta})K_{\epsilon_t}\big(\rho(\bm{s},\bm{s}_y)\big),
\]
in which the prior predictive is modulated by an $S$-only kernel. The rejuvenation kernels are designed to leave $\pi_t$ invariant, so the final population at tolerance $\epsilon_T$ approximates draws from a joint of the same form with $K_{\epsilon_T}$. 

This conditional invariance is specific to NPE. With NPE we model the posterior $q_{\bm{\phi}}(\bm{\theta}\mid\bm{s})$, and summary-only preconditioning leaves the target conditional $p(\bm{\theta}\mid\bm{s})$ unchanged on the kept support, so no correction is required. In contrast, neural likelihood estimation (NLE) models $q(\bm{s}\mid\bm{\theta})$. Reweighting the joint by \(w(\bm{s})\) changes the likelihood to
\[
p_w(\bm{s}\mid\bm{\theta}) \;\propto\; w(\bm{s})\,p(\bm{s}\mid\bm{\theta}),
\]
and an NLE model trained on the reweighted design would converge to this tilted likelihood unless one performs explicit correction (e.g.\ via importance weighting or re-simulation at evaluation time).

Summary-based preconditioning remains well defined even when the simulator or prior are misspecified. Because $S$-only weights $w_y(\bm{s})$ leave the model-implied conditional $\pi(\bm{\theta}\mid\bm{s})$ unchanged wherever $w_y(\bm{s})>0$, preconditioning does not introduce additional bias into the conditional target; it simply restricts attention to a data-dependent neighbourhood of $\bm{s}_y$. In random-forest preconditioning, tree-induced proximities act like an adaptive kernel in $\mathcal{S}$: simulations whose summaries fall in the same leaves as $\bm{s}_y$ receive large weights, while irrelevant regions are down-weighted. In SMC-ABC, the evolving particle system targets the standard ABC posterior, progressively pruning particles with large discrepancies $\rho(\bm{s},\bm{s}_y)$ so that mass concentrates near pseudo-true parameters whose simulated summaries resemble $\bm{s}_y$ \citep{frazier_model_2020}.

These schemes therefore direct computation towards a pseudo-true neighbourhood supported by the observed summaries and reduce the extrapolation risk associated with global regression adjustment. However, preconditioning alone cannot eliminate the extrapolation error $\Delta_{\mathrm{oos}}(\bm{s}_y)$: if $\bm{s}_y$ lies outside the effective support of $p(\bm{s}\mid\bm{\theta})$, the analysis remains extrapolative and inference is inherently unreliable. In our pipeline this motivates the use of robust NPE (RNPE) in addition to preconditioning, explicitly accounting for residual mismatch between the preconditioned training distribution and the observation.

\subsection{Preconditioned robust neural posterior estimation}\label{sec:prnpe_pipeline}
To address the dual challenges of poor prior-predictive behaviour and incompatible summaries, we couple preconditioning with misspecification-robust SBI. We propose preconditioned RNPE (PRNPE; Algorithm~\ref{alg:prnpe}).
Given a prior-predictive training set $\{(\bm{\theta}_i, \bm{s}_i)\}_{i=1}^N$, we compute normalised weights $\tilde{w}_i \propto w_y(\bm{s}_i)$ to concentrate computation near the observation.
We consider two choices for $w_y$ to filter out large-discrepancy simulations: (i) forest-proximities, which induce a tolerance-free neighbourhood around $\bm{s}_y$, and (ii) a short SMC-ABC warm-up \citep{wang_preconditioned_2024}.
The robust estimator is then trained on the resampled pairs following preconditioning.

\begin{algorithm}[ht]
\caption{PRNPE: preconditioned robust neural posterior estimation}
\label{alg:prnpe}
\begin{algorithmic}[1]
\Require Observed summary $\bm{s}_y$; prior $\pi(\bm{\theta})$; simulator $P_{\bm{\theta}}$; summary map $S$; simulation budget $N$; posterior sample size $M$; preconditioning weight function $w_y(\cdot)$; spike-and-slab hyperparameters $(\sigma_{\mathrm{spike}},\,\sigma_{\mathrm{slab}},\,\gamma)$.
\Ensure Posterior draws $\{\tilde{\bm{\theta}}_m\}_{m=1}^M$.
\State Draw $\bm{\theta}_i \sim \pi(\bm{\theta})$, simulate $\bm{x}_i \sim P_{\bm{\theta}_i}$, and set $\bm{s}_i \gets S(\bm{x}_i)$ for $i=1,\ldots,N$.
\State Compute weights $w_i \gets w_y(\bm{s}_i)$ via chosen preconditioning approach.
\State Compute weighted mean $\hat{\bm{\mu}}$ and standard deviation $\hat{\bm{\sigma}}$ from $\{(\bm{s}_i, \tilde{w}_i)\}_{i=1}^N$; standardise $\bm{s}_i \gets (\bm{s}_i - \hat{\bm{\mu}})/\hat{\bm{\sigma}}$ and $\bm{s}_y \gets (\bm{s}_y - \hat{\bm{\mu}})/\hat{\bm{\sigma}}$.
\State Normalise $\tilde w_i \gets w_i / \sum_{j=1}^N w_j$.
\State Fit conditional estimator $\hat{\bm{\phi}} \gets \arg\min_{\bm{\phi}} \sum_{i=1}^N \tilde w_i \big[-\log q_{\bm{\phi}}(\bm{\theta}_i\mid \bm{s}_i)\big]$.
\State Fit marginal density $\hat{\bm{\psi}} \gets \arg\max_{\bm{\psi}} \sum_{i=1}^N \tilde w_i \log h_{\bm{\psi}}(\bm{s}_i)$.
\State Define the spike-and-slab error model componentwise: $p(\bm{s}_y \mid \bm{s}) = \prod_{k=1}^{d_s}\big[(1-\gamma)\,\mathcal{N}(s_{y,k} \mid s_k,\,\sigma_{\mathrm{spike}}^2) + \gamma\,\mathcal{N}(s_{y,k} \mid s_k,\,\sigma_{\mathrm{slab}}^2)\big]$. Hyperparameter settings are given in Appendix~\ref{sec:hyperparameters}.
\For{$m=1,\ldots,M$}
\State Sample latent summary $\tilde{\bm{s}}_m \sim \widehat{p}(\bm{s}\mid\bm{s}_y) \propto p(\bm{s}_y\mid\bm{s})\, h_{\hat{\bm{\psi}}}(\bm{s})$ via NUTS; see Appendix~\ref{sec:hyperparameters} for sampler settings.
\State Sample parameter $\tilde{\bm{\theta}}_m \sim q_{\hat{\bm{\phi}}}(\bm{\theta}\mid \tilde{\bm{s}}_m)$.
\EndFor
\State \Return $\{\tilde{\bm{\theta}}_m\}_{m=1}^M$. 
\end{algorithmic}
\end{algorithm}

We adopt a forest-proximity preconditioning approach based on the ABC-RF framework \citep{raynal_abc_2019}. We fit a regression forest for each parameter component $\theta_j$ and define weights via leaf co-occurrence with the observation $\bm{s}_y$ \citep{meinshausen_quantile_2006},
 \[
w_y(\bm{s}_i) = \sum_{j=1}^{d_{\bm{\theta}}}\sum_{b=1}^B \frac{\mathbb{I}\{\bm{s}_i\in L_{jb}(\bm{s}_y)\}}{|L_{jb}(\bm{s}_y)|}.
\]
Since $w_y$ depends solely on $\bm{s}$, conditional invariance is satisfied.
However, unlike standard random forests which subsample features to decorrelate trees and reduce variance \citep{breiman_random_2001}, we include all summary features at every split.
While this may increase the variance of the parameter predictions, for preconditioning our priority is to aggressively isolate extreme simulations: using all features ensures that any coordinate exhibiting a large deviation from $\bm{s}_y$ can immediately trigger a split, excluding that simulation from the leaf $L_{jb}(\bm{s}_y)$.
We further enforce a conservative weighting by constraining tree depth and requiring large minimum leaf sizes; this coarsens the partition, smoothing the weights and maintaining a high effective sample size ($N_{\mathrm{eff}}$).
The full procedure is given in Algorithm~\ref{alg:rf-precond}.
Specific hyperparameters are detailed in Appendix~\ref{sec:hyperparameters}.

As an alternative preconditioning approach, we employ the SMC-ABC sampler of \citet{drovandi_estimation_2011}. Following \citet{wang_preconditioned_2024}, we run a limited number of generations to discard simulations with large discrepancies $\rho(\bm{s}, \bm{s}_y)$. This avoids the high computational cost of a full SMC run, effectively filtering the training data via a hard tolerance threshold without attempting to approximate the small-$\epsilon$ posterior.
Unlike \citet{wang_preconditioned_2024}, we do not fit an unconditional density to the ABC samples, instead using them to directly train the conditional estimator $q_{\bm{\phi}}(\bm{\theta}\mid\bm{s})$.
Specific settings for the population size and tolerance schedule are provided in Appendix~\ref{sec:app_smcabc_background}.

\begin{algorithm}[H]
\caption{Forest-proximity preconditioning}
\label{alg:rf-precond}
\begin{algorithmic}[1]
\Require Simulation dataset $\mathcal{D}=\{(\bm{\theta}_i,\bm{s}_i)\}_{i=1}^N$; observed summary $\bm{s}_y$; number of trees per parameter $B$; sample size $M$.
\Ensure Normalised weights $\tilde{\bm{w}}$ and resampled dataset $\mathcal{D}^{\mathrm{rf}}$.
\State Train $d_{\bm{\theta}}$ independent regression forests (one per parameter $\theta_j$), each containing $B$ trees.
\State For every tree $b \in \{1,\dots,B\}$ in every forest $j \in \{1,\dots,d_{\bm{\theta}}\}$, identify the leaf $L_{jb}(\bm{s}_y)$ containing $\bm{s}_y$.
\State Compute raw weights for each simulation $i$:
\[ w_i = \sum_{j=1}^{d_{\bm{\theta}}} \sum_{b=1}^B \frac{\mathbb{I}\{\bm{s}_i \in L_{jb}(\bm{s}_y)\}}{|L_{jb}(\bm{s}_y)|}. \]
\State Normalise weights: $\tilde{w}_i \gets w_i / \sum_{k=1}^N w_k$.
\State Resample with replacement: Draw indices $k_1,\dots,k_M \stackrel{\text{i.i.d.}}{\sim}\mathrm{Categorical}(\tilde{\bm{w}})$.
\State Construct dataset $\mathcal{D}^{\mathrm{rf}}=\{(\bm{\theta}_{k_m},\bm{s}_{k_m})\}_{m=1}^M$.
\State \Return Weights $\tilde{\bm{w}}$ and dataset $\mathcal{D}^{\mathrm{rf}}$.
\end{algorithmic}
\end{algorithm}



RNPE constitutes the robust inference stage of the pipeline. While preconditioning localises training near $\bm{s}_y$, RNPE handles potential incompatibility by denoising the observation via marginalisation over latent denoised summaries.
We train both the conditional estimator $q_{\bm{\phi}}(\bm{\theta}\mid \bm{s})$ and the marginal density estimator $h_{\bm{\psi}}(\bm{s})$ on the weighted dataset $\{(\bm{\theta}_i, \bm{s}_i, \tilde{w}_i)\}_{i=1}^N$ by maximising their respective weighted log-likelihoods.
Here, $h_{\bm{\psi}}$ targets the marginal density of the reweighted design $p_w(\bm{s}) \propto \int w_y(\bm{s})\,p_{\text{train}}(\bm{\theta},\bm{s})\,\mathrm{d}\bm{\theta}$, where $w_y(\bm{s})$ is determined by the chosen preconditioning scheme (random forest or SMC-ABC).
At inference, we form the approximate posterior over denoised summaries,
\begin{equation*}
\widehat{p}(\bm{s}\mid \bm{s}_y) \propto p(\bm{s}_y\mid \bm{s})\, h_{\bm{\psi}}(\bm{s}),
\end{equation*}
where $p(\bm{s}_y \mid \bm{s})$ is the likelihood induced by the spike-and-slab error model. We draw samples $\tilde{\bm{s}}_m \sim \widehat{p}(\cdot \mid \bm{s}_y)$ for $m=1, \ldots, M$, following the MCMC procedure of \citet{ward_robust_2022}, and compute the final posterior approximation via the ensemble average
\begin{equation*}
\hat{\pi}(\bm{\theta}\mid \bm{s}_y) \approx \frac{1}{M}\sum_{m=1}^M q_{\bm{\phi}}(\bm{\theta}\mid \tilde{\bm{s}}_m).
\end{equation*}
This two-stage approach ensures that preconditioning focuses modelling capacity on the relevant summary space, while RNPE allows robust inference under model incompatibility.

Under correct specification, RNPE recovers the standard NPE posterior and exhibits consistent concentration around the true parameter.
Conceptually, RNPE encourages towards compatible summaries by construction: by conditioning on $\bm{s}_y$ and sampling latent summaries $\tilde{\bm{s}}$ from the simulator's support, it effectively projects the observation back onto the model manifold via coordinate-wise probabilistic shifts.
RNPE targets a convolved summary likelihood $p_{\tau}(\bm{s}_y \mid \bm{\theta}) = \int p_{\tau}(\bm{s}_y \mid \bm{s}) p(\bm{s} \mid \bm{\theta}) \, \mathrm{d}\bm{s}$, where $p_{\tau}(\bm{s}_y \mid \bm{s})$ represents the chosen error model (we use the spike-and-slab specification
of \citet{ward_robust_2022} throughout) parameterised by $\tau$. Consequently, the posterior concentrates at the parameter $\bm{\theta}_{\tau}^\star$ that minimises the Kullback-Leibler divergence from the true DGP to this convolved model.
In contrast, RNLE explicitly models the additive adjustment $\bm{\Gamma}$. The joint posterior concentrates on the manifold $\mathcal{M} = \{(\bm{\theta}, \bm{\Gamma}) : b(\bm{\theta}) + \bm{\Gamma} = b_{\star}\}$, effectively finding the projection of $b_{\star}$ onto the model family $\{b(\bm{\theta})\}$ that minimises the penalty induced by the prior.

\section{Examples}\label{sec:examples}

We evaluate our preconditioning strategies on three benchmarks: two synthetic tasks and one real-world application. These examples are selected to stress-test inference under incompatible summaries and heavy-tailed prior-predictive distributions.
Across all cases, combining preconditioning with robust NPE significantly improves estimation accuracy and calibration compared to existing methods.
In the synthetic experiments, we introduce misspecification while retaining the underlying generative structure for relevant summaries. Consequently, the target pseudo-truth $\bm{\theta}^{\star}$ is the parameter generating the compatible components of the observation.
We assess performance using estimation accuracy (Bias, RMSE), uncertainty quantification (empirical coverage of 95\% highest posterior density intervals; HPDI), and the log-density of the true parameter. Additionally, we evaluate the posterior-predictive fit by computing the distance between simulated and observed summaries. To ensure a fair comparison, all methods are restricted to a fixed budget of 20,000 model simulations per dataset.

\subsection{Contaminated Weibull}

We illustrate the dual challenges of extreme prior-predictive behaviour and summary incompatibility using a toy one-parameter Weibull model.
The assumed DGP is
\begin{equation*}
      x_i \overset{\text{i.i.d.}}{\sim} \textsf{Weibull}(k, \lambda=1), \quad i=1,\ldots,n, \quad n = 200,
\end{equation*}
with unknown shape $k>0$ and fixed scale.
We assign a weakly informative prior $k \sim \text{LogNormal}(1, 1)$.
Since low shape parameters induce heavy-tailed distributions, prior-predictive summaries (particularly the sample variance) span many orders of magnitude, creating a challenging dataset for training.
We consider three summary statistics: the sample mean, the sample variance, and the minimum.
To introduce incompatibility, the true DGP is a mixture of the Weibull distribution and a $5\%$ contamination from $\mathcal{N}(-1, 0.2)$.
This contamination generates negative values, rendering the minimum summary incompatible with the strictly positive support of the assumed model.
We define the target pseudo-truth $k^{\star} \approx 0.789$ as the parameter that minimises the Euclidean distance between the expected model summaries and the expected contaminated summaries, restricting the objective to the compatible subset $(\bar{x}, s^2)$ (computed via grid search over analytic population moments).

Table~\ref{tab:contaminated_weibull_table} summarises the estimation accuracy and predictive performance for the contaminated Weibull example, averaged across 100 paired Monte Carlo replicates.
The forest-proximity approach yields slightly stronger accuracy and predictive fit, with lower RMSE (0.07 vs 0.09) and a better median posterior predictive distance (-0.53 vs -0.62). Conversely, the SMC-ABC variant attains near-nominal coverage (0.98 vs 0.85), highlighting a mild trade-off between sharpness and calibration. In contrast, the baselines fail to cope with the combination of incompatibility and extreme prior-predictive draws despite utilising a comparable computational budget. Standard NPE and RNPE exhibit negligible coverage and poor predictive performance, while preconditioned NPE (PNPE) without the denoising step suffers from higher estimation error and poorer calibration than the robust counterparts. These results demonstrate that combining preconditioning with robustness is essential to recover accuracy and predictive fit while maintaining reliable calibration.

\begin{table}[ht]
\centering
\small
\caption{
Performance comparison for the contaminated Weibull example ($n=200$). Values denote means (standard deviations) across 100 paired Monte Carlo replicates.
Metrics include estimation error (Bias, RMSE), uncertainty calibration (empirical coverage of $k^\star$ by 95\% HPDIs), and predictive fit.
Predictive fit is reported as the posterior predictive distance (PPD), defined as the median log-Euclidean distance between posterior simulations and the observation. Arrows indicate the desirable direction.
}
\label{tab:contaminated_weibull_table}

\begin{tabular}{
l
l
l
l
l                    
l                    
}
\toprule
& \multicolumn{1}{c}{Bias $\downarrow$}
& \multicolumn{1}{c}{RMSE $\downarrow$}
& \multicolumn{1}{c}{Coverage}
& \multicolumn{1}{c}{$\log$ PPD (median) $\downarrow$} \\
\midrule

\textsc{prnpe (smc abc)}        & \textbf{0.05 (0.04)} & 0.09 (0.02) & \textbf{0.98}  & -0.53 (0.32) \\
\textsc{prnpe (forest-proximity)}  &   \textbf{0.05 (0.04)} & \textbf{0.07 (0.03)}  & 0.85          & \textbf{-0.62 (0.29) } \\
\textsc{npe}          & 1.01 (2.28)  &  1.01 (2.28)    & 0.00          & 92.30 (131.53) \\
\textsc{rnpe}         &  4.60 (12.99)  & 5.29 (14.19)         & 0.30          &  1.10 (74.89) \\
\textsc{pnpe (smc abc)}         &  0.38 (0.26) & 0.46 (0.27) & 0.17          & 3.27 (2.32) \\
\textsc{pnpe (forest-proximity)}   &  0.33 (0.23) & 0.40 (0.22)          & 0.19          & 0.85 (1.99) \\
\bottomrule
\end{tabular}
\end{table}

Figure~\ref{fig:weibull_posteriors} displays the marginal posteriors for the Weibull shape $k$ obtained via PRNPE for a particular dataset. Under both preconditioning schemes, the posterior mass concentrates around the pseudo-true target $k^{\star}$, avoiding the implausible regions induced by the heavy-tailed prior.
Figure~\ref{fig:weibull_ppc} presents posterior predictive checks for the compatible summaries $(\bar{x}, s^2)$.
The predictive distributions concentrate tightly around the observed statistics, with the observation falling well within the 50\% and 90\% HPD regions.
Crucially, unlike the extreme prior-predictive distribution, the posterior simulations remain within a realistic range of the summary space.

\begin{figure}
    \centering
    \includegraphics[width=0.99\linewidth]{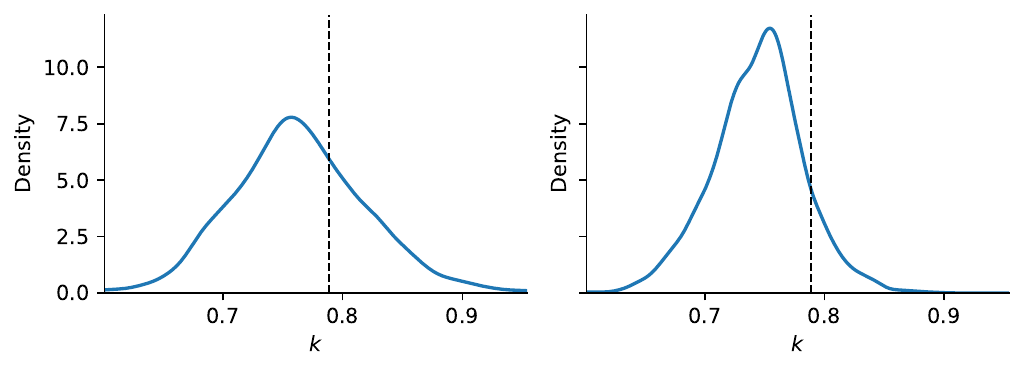}
    \caption{Posterior densities for the Weibull shape $k$ using PRNPE with SMC-ABC (left) and forest-proximity (right) preconditioning. The vertical dashed line indicates the target pseudo-truth $k^{\star}$.}
    \label{fig:weibull_posteriors}
\end{figure}

\begin{figure}
    \centering
    \includegraphics[width=0.99\linewidth]{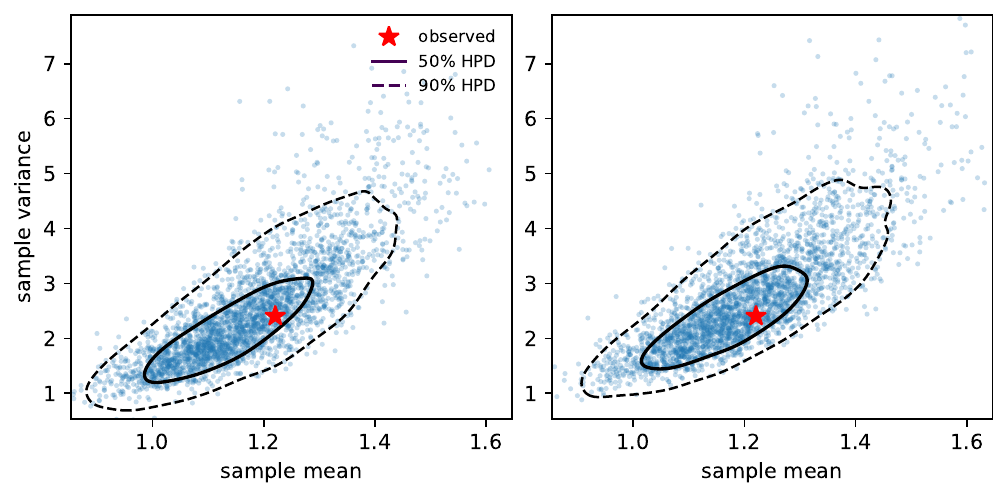}
    \caption{Posterior predictive distributions for the compatible summaries $(\bar{x}, s^2)$ in the contaminated Weibull example using PRNPE with SMC-ABC (left) and forest-proximity (right) preconditioning. Scatter points represent 2000 draws; contours indicate the 50\% (solid) and 90\% (dashed) highest-density regions. The red star denotes the observed summary $\bm{s}_y$.}
    \label{fig:weibull_ppc}
\end{figure}

\subsection{Sparse vector autoregression}
We consider a sparse vector autoregression (SVAR) model, previously used to benchmark SBI methods \citep{drovandi_improving_2024}, which exhibits extreme prior-predictive behaviour that motivates preconditioning \citep{wang_preconditioned_2024}.
We set the dimension $d=6$ and define the transition matrix $A \in \R^{d \times d}$ with diagonal elements fixed at $-0.1$. Non-zero off-diagonals are restricted to three disjoint pairs $(1,2), (3,4), (5,6)$, where $A_{i,j}$ and $A_{j,i}$ are estimated separately.
The assumed model is
\begin{equation*}
\bm{y}_t = A\bm{y}_{t-1} + \bm{\xi}_t, \quad t=1,\ldots,T, \quad T=1000, \quad \bm{\xi}_t \sim \mathcal{N}(0, \sigma^2 I_d).
\end{equation*}
We place independent uniform priors $A_{i, j} \sim \mathcal{U}(-1, 1)$ on the six active off-diagonals and $\sigma \sim \mathcal{U}(0,1)$ on the noise scale.
For each run, we generate an observed sequence using $\bm{\theta} = (0.579, -0.143, 0.836, 0.745, -0.660, -0.254, 0.1)$.
We compute eight summary statistics: the six lag-1 cross-covariances $1/T \sum_{t=2}^T (y_{i,t} - \bar{y}_i) (y_{j,t-1} - \bar{y}_j)$ corresponding to the active pairs, the pooled standard deviation, and the global mean.
To introduce misspecification, we add a constant drift $\mu$ to each component with $\mu=0.05$ to the generative process.
While this drift renders the mean summary incompatible with the zero-mean assumption of the model, the centring in the covariance and standard deviation calculations renders the remaining seven summaries compatible. Consequently, the target pseudo-truth matches the data-generating parameters.

Table~\ref{tab:svar_table} summarises performance for the SVAR model under drift misspecification, averaged across 100 paired Monte Carlo replicates.
PRNPE (SMC-ABC) performs best overall, achieving very accurate inference on $\sigma$ (Bias 0.00, RMSE 0.00, Coverage 1.00) and the strongest predictive fit (median PPD -4.15). The forest-proximity variant is nearly as accurate (Bias 0.01, RMSE 0.02) and maintains near-nominal coverage (0.98) with a comparable predictive fit (-4.06). While standard NPE attains the highest log-probability at the pseudo-truth, it suffers from large estimation errors and weak predictive fit. Similarly, RNPE and the non-robust preconditioned baselines exhibit higher errors, weaker predictive fit, and poor calibration.

\begin{table}[ht]
\centering
\small
\caption{Performance comparison for the sparse VAR(1) example ($d=6, T=1000$). Values denote means (standard deviations) across 100 paired Monte Carlo replicates. Metrics include marginal estimation error for $\sigma$ (Bias, RMSE), uncertainty calibration (empirical coverage of $\sigma^\star$ by 95\% HPDIs), target fit (amortised log-density at the pseudo-truth $\bm{\theta}^\star$), and predictive fit (log-median Posterior Predictive Distance; PPD). The PPD is computed using only the seven compatible summaries (six lag-1 cross-covariances and the pooled s.d.). Arrows indicate the direction of improvement.}
\label{tab:svar_table}
\begin{tabular}{lrrrr}
\hline
& \multicolumn{3}{c}{Marginal $(\sigma)$} & \multicolumn{1}{c}{Model-wide} \\
Method & Bias $\downarrow$ & RMSE $\downarrow$ & Coverage &
$\log \mathrm{PPD}\, \text{(median)}\,\downarrow$ \\
\hline
\textsc{prnpe (smc abc)} & \textbf{0.00} (0.00) & \textbf{0.00} (0.00) & 1  & \textbf{-4.15 (0.17)} \\
\textsc{prnpe (forest-proximity)}       & 0.01 (0.02) & 0.02 (0.02) &  \textbf{0.98 }          & -4.06(7.12)\\
\textsc{npe}        & 0.40 (0.02) & 0.49 (0.01)          & 1   &  -0.24 (0.10)  \\
\textsc{rnpe}  & 0.38 (0.22)          &  0.48 (0.18)        & 0.76          &  0.44 (24.89)\\
\textsc{pnpe (smc abc)}         & 0.06 (0.08)          & 0.07 (0.08)     &  0.15           & 8.22 (85.38) \\
\textsc{pnpe (forest-proximity)}        & 0.10 (0.12)         & 0.14 (0.13)          & 0.50          & 86.59 (80.51) \\
\hline
\end{tabular}
\end{table}

Figure~\ref{fig:svar_sigma_posteriors} displays the marginal posteriors for the noise scale $\sigma$ under PRNPE with SMC-ABC and forest-proximity preconditioning. Both variants concentrate tightly near the data-generating value $\sigma^\star=0.10$.
Posterior predictive checks for the compatible pooled standard deviation (Figure~\ref{fig:svar_ppc}) show that the simulated statistics cluster around the observation, with the observed value falling within the interquartile range of the posterior predictive distributions.
These diagnostics corroborate the strong predictive performance reported in Table~\ref{tab:svar_table}.

\begin{figure}[htbp]
  \centering
  \begin{subfigure}[t]{0.49\textwidth}
    \centering
    \includegraphics[width=\linewidth]{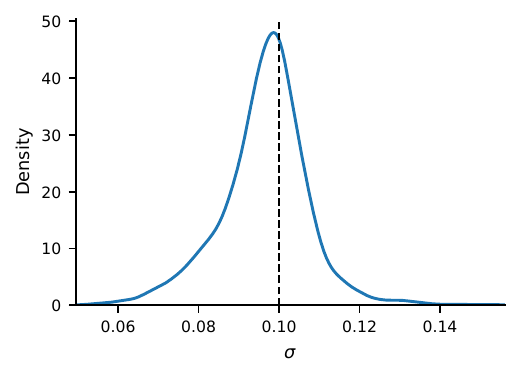}
    \caption{PRNPE (SMC-ABC)}
    \label{fig:svar_sigma_prnpe}
  \end{subfigure}
  \hfill
  \begin{subfigure}[t]{0.49\textwidth}
    \centering
    \includegraphics[width=\linewidth]{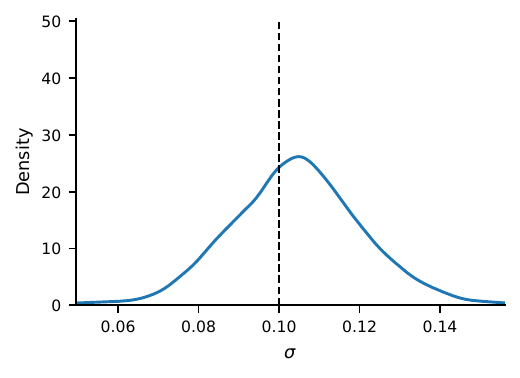}
    \caption{PRNPE (random forest)}
    \label{fig:svar_sigma_rfabc}
  \end{subfigure}
\caption{Marginal posterior densities for the noise scale $\sigma$ in the SVAR example using PRNPE with SMC ABC (left) and forest-proximity (right) preconditioning. The vertical dashed line indicates the target pseudo-truth $\sigma^\star$.}
  \label{fig:svar_sigma_posteriors}
\end{figure}

\begin{figure}
    \centering
    \includegraphics[width=0.5\linewidth]{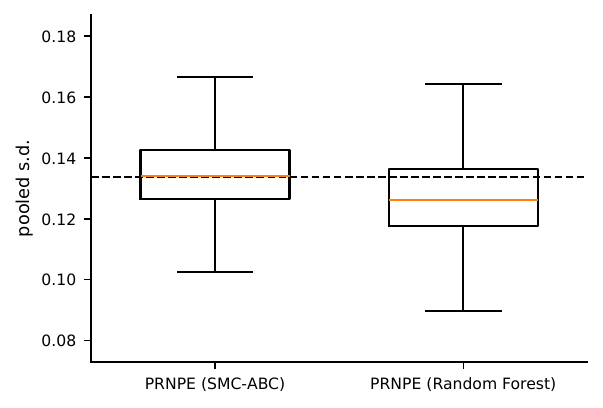}
\caption{Posterior predictive distributions for the pooled standard deviation in the SVAR example using PRNPE with SMC ABC (left) and forest-proximity (right) preconditioning. The dashed horizontal line indicates the observed summary statistic.}
    \label{fig:svar_ppc}
\end{figure}

\subsection{Real data: BVCBM}\label{sec:bvcbm}

We calibrate a biphasic Voronoi cell-based model (BVCBM) to four pancreatic tumour time series consisting of daily volume measurements ($\mathrm{mm}^3$) with durations $T\in\{19,26,32,32\}$ days \citep{wade_dual_2020}; we fit data from each mouse independently.
The model assumes a biphasic process with a growth rate switch at day $\tau$.
We estimate the reduced parameter vector $\bm{\theta}=(g_{\mathrm{age}}^{(1)}, \tau, g_{\mathrm{age}}^{(2)})$, where $g_{\mathrm{age}}$ denotes the minimum re-division time (in hours) and $\tau$ is measured in days.
To ensure identifiability given volume-only data, we fix the remaining parameters $(p_0,p_{\mathrm{psc}},d_{\max},p_{\text{age}})=(1.0,\,0,\,17.3,\,5)$.
Rather than reducing the data to lower-dimensional statistics, we utilise the full volume trajectories as summaries and compute discrepancies via Euclidean distance.

We assign weakly informative uniform priors: $g_{\mathrm{age}}^{(k)}\sim\mathcal{U}(2,\,24T{-}1)$ for $k\in\{1,2\}$ and $\tau\sim\mathcal{U}(2,\,T{-}1)$.
While physiologically plausible, these ranges are broad. Low $g_{\mathrm{age}}$ values induce rapid growth while high values yield near-stasis, and varying $\tau$ shifts the phase transition; consequently, the prior-predictive distribution covers a wide range of trajectory shapes.
Misspecification arises primarily from measurement noise and discrepancies in initial conditions.

Since the scientific object of interest is the tumour volume trajectory, we assess predictive fit using the PPD computed on the full time series. Table~\ref{tab:bvcbm_ppd_distance} shows that PRNPE (forest-proximity) achieves the strongest fit, attaining the lowest PPD on datasets D1--D3 (0.367, 0.300, 0.328) and tying for best on D4 (0.433). The SMC-ABC variant remains competitive but yields slightly higher distances. In contrast, non-robust baselines exhibit instability: NPE degrades on D3 and D4 (1.147, 1.201), while RNPE performs poorly on D2 and D3 (1.704, 1.472) despite being competitive on D1. The PNPE baselines are consistently outperformed, achieving parity only on D4. Overall, combining preconditioning with robustness yields the best posterior-predictive trajectory fit across the four datasets.

\begin{table}[t]
\centering
\caption{PPD for the four tumour-growth datasets. Values represent the median Euclidean distance between simulated and observed trajectories; lower values indicate better fit.}
\label{tab:bvcbm_ppd_distance}
\small
\setlength{\tabcolsep}{6pt}
\begin{tabular}{l c c c c}
\toprule
Method & {D1} & {D2} & {D3} & {D4} \\
\midrule
\textsc{prnpe (smc abc)} & 0.561 & 0.437 & 0.397 & 0.571 \\
\textsc{prnpe (forest-proximity)} & \textbf{0.367} & \textbf{0.300} & \textbf{0.328} & \textbf{0.433} \\
\textsc{npe} & 0.498 & 0.373 & 1.147 & 1.201 \\
\textsc{rnpe} & 0.370 & 1.704 & 1.472 & 0.478 \\
\textsc{pnpe (smc abc)} & 0.674 & 0.647 & 0.833 & 0.596 \\
\textsc{pnpe (forest-proximity)} & 0.414 & 1.113 & 0.352 & \textbf{ 0.433} \\
\bottomrule
\end{tabular}

\vspace{2pt}
\begin{minipage}{0.96\linewidth}\footnotesize
\end{minipage}
\end{table}

\section{Discussion}\label{sec:discussion}

Simulation-based Bayesian inference can be brittle when simulators behave poorly under broad priors or when models are misspecified. We address these dual failure modes with preconditioned robust neural posterior estimation (PRNPE).
While prior work has addressed either training stability via preconditioning in well-specified settings \citep{wang_preconditioned_2024} or incompatibility via robust adjustments \citep{ward_robust_2022, kelly_misspecification-robust_2024}, PRNPE addresses both. Further, our preconditioning simplifies the pipeline of \citet{wang_preconditioned_2024} by training the conditional estimator directly on the preconditioned samples rather than requiring an intermediate unconditional fitting step.
By preconditioning the training around the observed summary, via either our proposed forest-proximity approach or a pilot SMC-ABC run, and subsequently pairing this with a robust SBI method that can mitigate incompatible summaries, our approach focuses the neural network's model capacity on the relevant parameter region, even under incompatible summaries. Across synthetic studies and a real-data example, we found this combination yielded stable training and stronger posterior-predictive fit than baseline methods. 

The effectiveness of PRNPE stems from its ability to reduce the amortisation gap without biasing the inference. Standard amortised NPE minimises a global loss over the prior-predictive distribution; when this distribution spans orders of magnitude, the network wastes capacity modelling irrelevant regions, degrading accuracy at the specific observation.
Crucially, because our preconditioning relies solely on summaries, it leaves the conditional target invariant, avoiding the target-shift bias associated with regression-adjustment ABC. However, preconditioning alone cannot create support where the simulator has none. In cases of incompatibility, preconditioning filters towards the closest possible simulations under the assumed model, while the subsequent robust SBI method bridges the remaining gap between the simulated manifold and the incompatible observation. Furthermore, because robust methods like RNPE include explicit error terms (e.g., spike-and-slab indicators), they allow for model criticism, enabling the modeller to diagnose which summaries are driving the incompatibility, allowing for subsequent model refinement.

PRNPE offers a practical option for SBI in the early stages of a Bayesian workflow, where weakly informative priors are valuable but often result in extreme prior-predictive behaviour. While tightening the prior can mitigate extremes, it shifts the burden onto elicitation and requires costly retuning. In contrast, PRNPE stabilises training without redesigning the prior, maintaining model transparency. Compared to classical ABC, which handles extremes by rejection but requires prohibitive simulation budgets, PRNPE retains the efficiency of the core NPE architecture, requiring only a modest number of additional simulations for the preconditioning step. As a general guideline, we recommend preconditioning when prior-predictive plots span many orders of magnitude. We recommend defaulting to conservative weighting that preserves breadth, with the main benefit coming from removing implausible regions that degrade training without discarding simulations that may still be useful for generalisation.

We focused on single-round (amortised) NPE to isolate the quality of the initial approximation, but our results have direct implications for sequential methods. Sequential schemes rely on the first round to generate a valid proposal; if the initial posterior is poor due to extreme priors, the subsequent rounds may fail to converge. PRNPE provides the high-quality initial proposal to start sequential inference. Regarding the preconditioning implementations, we note a trade-off: the forest-proximity approach is computationally convenient as it runs on the initial prior-predictive draws without intermediate simulation loops, but it can produce spiky weights with low effective sample size (ESS). In our setup, we used all features at each split to ensure extreme summary dimensions were captured; while this reduced tree diversity and lowered ESS, performance remained strong in our examples.

Future research could extend this framework to ``semi-amortised'' inference across a dataset, redefining the preconditioning weight to cover the union of neighbourhoods around all observed summaries. This would focus neural capacity on the simulator's effective operational domain while avoiding the learning of conditional mappings in implausible prior-predictive regions. Additionally, to mitigate spiky weights in the forest-proximity approach, one could combine it with SMC-ABC as recently done by \citet{dinh_approximate_2025}, using the ABC component to filter extremes and enabling feature subsampling to restore tree decorrelation. Finally, while we utilised RNPE, alternative robust methods such as learnt robust summaries \citep{huang_learning_2023} or post-hoc optimal-transport calibration \citep{wehenkel_addressing_2025} could be substituted as the second stage.
Complementary strategies for improving amortisation quality include self-consistency regularisation \citep{schmitt_leveraging_2024}, which penalises violations of the marginal identity $\int q_{\bm{\phi}}(\bm{\theta}\mid\bm{s})\,p(\bm{s})\,\mathrm{d}\bm{s} = \pi(\bm{\theta})$, with recent work extending this idea to unlabelled observed data for robustness to distribution shift \citep{mishra_robust_2025}. Preconditioning addresses the same failure mode from the loss-design side, and one could combine it with self-consistency penalties.
Ultimately, preconditioning around the observation, paired with a robust SBI method, provides a practical approach for principled inference even under challenging misspecified simulators and wide priors.






\section*{Acknowledgements}
This work was supported by the Australian Government Research Training
Program (scholarship to RPK); the QUT Centre for Data Science (top-up
scholarship to RPK); the Australian Research Council under Future
Fellowship FT210100260 (to CD) and Discovery Early Career Researcher
Award DE250100396 (to DJW). CD and DJW also acknowledge support from
the ARC Centre of Excellence for the Mathematical Analysis of Cellular
Systems (MACSYS).
Computational resources and services were provided by the HPC and
Research Support Group, Queensland University of Technology.

\section*{Data availability}
The pancreatic tumour volume data used in Section~\ref{sec:bvcbm}
were originally collected by \citet{wade_dual_2020}. All remaining
data are synthetically generated. Code to reproduce all experiments
is available at
\url{https://github.com/RyanJafefKelly/preconditioned-npe-under-misspecification}.





\bibliography{references}

@article{wade_dual_2020,
author = {Wade, Samantha J. and Sahin, Zeliha and Piper, Ann-Katrin and Talebian, Sepehr and Aghmesheh, Morteza and Foroughi, Javad and Wallace, Gordon G. and Moulton, Simon E. and Vine, Kara L.},
title = {Dual delivery of Gemcitabine and Paclitaxel by wet-spun coaxial fibers induces pancreatic ductal Adenocarcinoma cell death, reduces tumor volume, and sensitizes cells to radiation},
journal = {Advanced Healthcare Materials},
volume = {9},
number = {21},
pages = {2001115},
doi = {https://doi.org/10.1002/adhm.202001115},
eprint = {https://advanced.onlinelibrary.wiley.com/doi/pdf/10.1002/adhm.202001115},
year = {2020}
}

@inproceedings{lueckmann_flexible_2017,
	title = {Flexible statistical inference for mechanistic models of neural dynamics},
	volume = {30},
	urldate = {2024-05-01},
	booktitle = {Advances in {Neural} {Information} {Processing} {Systems}},
	author = {Lueckmann, Jan-Matthis and Gonçalves, Pedro J and Bassetto, Giacomo and Öcal, Kaan and Nonnenmacher, Marcel and Macke, Jakob H},
	year = {2017},
}

@inproceedings{lueckmann_benchmarking_2021,
	title = {Benchmarking simulation-based inference},
	volume = {130},
	language = {en},
	urldate = {2022-08-03},
	booktitle = {Proceedings of the 24th {International} {Conference} on {Artificial} {Intelligence} and {Statistics}},
	publisher = {PMLR},
	author = {Lueckmann, Jan-Matthis and Boelts, Jan and Greenberg, David and Gonçalves, Pedro and Macke, Jakob},
	month = mar,
	year = {2021},
	pages = {343--351},
}

@inproceedings{greenberg_automatic_2019,
	title = {Automatic posterior transformation for likelihood-free inference},
	volume = {97},
	language = {en},
	urldate = {2022-07-29},
	booktitle = {Proceedings of the 36th {International} {Conference} on {Machine} {Learning}},
	publisher = {PMLR},
	author = {Greenberg, David and Nonnenmacher, Marcel and Macke, Jakob},
	month = may,
	year = {2019},
	pages = {2404--2414},
}

@article{wang_preconditioned_2024,
	title = {Preconditioned neural posterior estimation for likelihood-free inference},
	issn = {2835-8856},
	journal = {Transactions on Machine Learning Research},
	author = {Wang, Xiaoyu and Kelly, Ryan P. and Warne, David J. and Drovandi, Christopher},
	year = {2024},
	note = {\url{https://openreview.net/forum?id=vgIBAOkIhY}},
}

@book{sisson_handbook_2018,
	address = {New York},
	title = {Handbook of approximate {Bayesian} computation},
	isbn = {978-1-315-11719-5},
	doi = {10.1201/9781315117195},
	publisher = {Chapman and Hall/CRC},
	editor = {Sisson, Scott A. and Fan, Yanan and Beaumont, Mark},
	month = sep,
	year = {2018},
}

@inproceedings{schmitt_leveraging_2024,
	title = {Leveraging self-consistency for data-efficient amortized {Bayesian} inference},
	volume = {235},
	urldate = {2025-05-10},
	booktitle = {Proceedings of the 41st {International} {Conference} on {Machine} {Learning}},
	publisher = {PMLR},
	author = {Schmitt, Marvin and Ivanova, Desi R. and Habermann, Daniel and Köthe, Ullrich and Bürkner, Paul-Christian and Radev, Stefan T.},
	month = jul,
	year = {2024},
	pages = {43723--43741},
}

@inproceedings{rezende_variational_2015,
	title = {Variational inference with normalizing flows},
	volume = {37},
	language = {en},
	urldate = {2024-11-30},
	booktitle = {Proceedings of the 32nd {International} {Conference} on {Machine} {Learning}},
	publisher = {PMLR},
	author = {Rezende, Danilo and Mohamed, Shakir},
	month = jun,
	year = {2015},
	pages = {1530--1538},
}

@inproceedings{papamakarios_sequential_2019,
	title = {Sequential neural likelihood: fast likelihood-free inference with autoregressive flows},
	volume = {89},
	shorttitle = {Sequential {Neural} {Likelihood}},
	language = {en},
	urldate = {2022-07-29},
	booktitle = {Proceedings of the {Twenty}-{Second} {International} {Conference} on {Artificial} {Intelligence} and {Statistics}},
	publisher = {PMLR},
	author = {Papamakarios, George and Sterratt, David and Murray, Iain},
	month = apr,
	year = {2019},
	pages = {837--848},
}

@misc{lyu_dynamic_2025,
	title = {Dynamic {SBI}: round-free sequential simulation-based inference with adaptive datasets},
	shorttitle = {Dynamic {SBI}},
	doi = {10.48550/arXiv.2510.13997},
	urldate = {2025-10-28},
	publisher = {arXiv},
	author = {Lyu, Huifang and Alvey, James and Montel, Noemi Anau and Pieroni, Mauro and Weniger, Christoph},
	month = oct,
	year = {2025},
	note = {arXiv preprint arXiv:2510.13997},
}

@misc{kingma_adam_2017,
	title = {Adam: a method for stochastic optimization},
	shorttitle = {Adam},
	doi = {10.48550/arXiv.1412.6980},
	urldate = {2023-11-15},
	publisher = {arXiv},
	author = {Kingma, Diederik P. and Ba, Jimmy},
	month = jan,
	year = {2017},
	note = {arXiv preprint arXiv:1412.6980},
}

@inproceedings{hermans_likelihood-free_2020,
	title = {Likelihood-free {MCMC} with amortized approximate ratio estimators},
	volume = {119},
	language = {en},
	urldate = {2022-08-04},
	booktitle = {Proceedings of the 37th {International} {Conference} on {Machine} {Learning}},
	publisher = {PMLR},
	author = {Hermans, Joeri and Begy, Volodimir and Louppe, Gilles},
	month = nov,
	year = {2020},
	pages = {4239--4248},
}

@inproceedings{durkan_neural_2019,
	title = {Neural spline flows},
	volume = {32},
	booktitle = {Advances in {Neural} {Information} {Processing} {Systems}},
	publisher = {Curran Associates, Inc.},
	author = {Durkan, Conor and Bekasov, Artur and Murray, Iain and Papamakarios, George},
	editor = {Wallach, H. and Larochelle, H. and Beygelzimer, A. and Alché-Buc, F. d' and Fox, E. and Garnett, R.},
	year = {2019},
}

@inproceedings{durkan_contrastive_2020,
	title = {On contrastive learning for likelihood-free inference},
	volume = {119},
	language = {en},
	urldate = {2022-08-04},
	booktitle = {Proceedings of the 37th {International} {Conference} on {Machine} {Learning}},
	publisher = {PMLR},
	author = {Durkan, Conor and Murray, Iain and Papamakarios, George},
	month = nov,
	year = {2020},
	pages = {2771--2781},
}

@inproceedings{cremer_inference_2018,
	title = {Inference suboptimality in variational autoencoders},
	volume = {80},
	language = {en},
	urldate = {2025-10-22},
	booktitle = {Proceedings of the 35th {International} {Conference} on {Machine} {Learning}},
	publisher = {PMLR},
	author = {Cremer, Chris and Li, Xuechen and Duvenaud, David},
	month = jul,
	year = {2018},
	pages = {1078--1086},
}

@inproceedings{beck_efficient_2022,
	address = {Red Hook, NY, USA},
	series = {{NeurIPS} 2022},
	title = {Efficient identification of informative features in simulation-based inference},
	isbn = {978-1-7138-7108-8},
	urldate = {2026-01-06},
	booktitle = {Proceedings of the 36th {International} {Conference} on {Neural} {Information} {Processing} {Systems}},
	publisher = {Curran Associates Inc.},
	author = {Beck, Jonas and Deistler, Michael and Bernaerts, Yves and Macke, Jakob H. and Berens, Philipp},
	month = nov,
	year = {2022},
	pages = {19260--19273},
}

@article{kelly_misspecification-robust_2024,
	title = {Misspecification-robust sequential neural likelihood for simulation-based inference},
	issn = {2835-8856},
	journal = {Transactions on Machine Learning Research},
	author = {Kelly, Ryan P. and Nott, David J. and Frazier, David Tyler and Warne, David J. and Drovandi, Christopher},
	year = {2024},
	note = {\url{https://openreview.net/forum?id=tbOYJwXhcY}},
}

@article{breiman_random_2001,
	title = {Random forests},
	volume = {45},
	issn = {1573-0565},
	doi = {10.1023/A:1010933404324},
	language = {en},
	number = {1},
	urldate = {2025-11-06},
	journal = {Machine Learning},
	author = {Breiman, Leo},
	month = oct,
	year = {2001},
	pages = {5--32},
}

@article{bi_random_2022,
	title = {Random forest adjustment for approximate {Bayesian} computation},
	volume = {31},
	doi = {10.1080/10618600.2021.1981341},
	number = {1},
	journal = {Journal of Computational and Graphical Statistics},
	author = {Bi, Jiefeng and Shen, Weining and Zhu, Weixuan},
	year = {2022},
	pages = {64--73},
}

@article{beaumont_adaptive_2009,
	title = {Adaptive approximate {Bayesian} computation},
	volume = {96},
	issn = {0006-3444},
	doi = {10.1093/biomet/asp052},
	number = {4},
	urldate = {2025-10-21},
	journal = {Biometrika},
	author = {Beaumont, Mark A. and Cornuet, Jean-Marie and Marin, Jean-Michel and Robert, Christian P.},
	year = {2009},
	pages = {983--990},
}

@inproceedings{phan_composable_2019,
	title = {Composable effects for flexible and accelerated probabilistic programming in {NumPyro}},
	language = {en},
	urldate = {2025-12-21},
	author = {Phan, Du and Pradhan, Neeraj and Jankowiak, Martin},
	month = sep,
	year = {2019},
	note = {\url{https://github.com/pyro-ppl/numpyro}},
}

@misc{ward_flowjax_2025,
	title = {{FlowJAX}: distributions and normalizing flows in {Jax}},
	author = {Ward, Daniel},
	year = {2025},
	note = {\url{https://github.com/danielward27/flowjax}},
}

@article{drovandi_improving_2024,
	title = {Improving the accuracy of marginal approximations in likelihood-free inference via localization},
	volume = {33},
	issn = {1061-8600},
	doi = {10.1080/10618600.2023.2223574},
	number = {1},
	urldate = {2024-06-13},
	journal = {Journal of Computational and Graphical Statistics},
	author = {Drovandi, Christopher C. and Nott, David J. and Frazier, David T.},
	month = jan,
	year = {2024},
	pages = {101--111},
}

@article{drovandi_estimation_2011,
	title = {Estimation of parameters for macroparasite population evolution using approximate {Bayesian} computation},
	volume = {67},
	issn = {0006-341X},
	doi = {10.1111/j.1541-0420.2010.01410.x},
	number = {1},
	urldate = {2024-05-08},
	journal = {Biometrics},
	author = {Drovandi, Christopher C. and Pettitt, A. N.},
	month = mar,
	year = {2011},
	pages = {225--233},
}

@article{pedregosa_scikit-learn_2011,
	title = {Scikit-learn: machine learning in {Python}},
	volume = {12},
	issn = {1533-7928},
	shorttitle = {Scikit-learn},
	number = {85},
	urldate = {2025-11-28},
	journal = {Journal of Machine Learning Research},
	author = {Pedregosa, Fabian and Varoquaux, Gael and Gramfort, Alexandre and Michel, Vincent and Thirion, Bertrand and Grisel, Olivier and Blondel, Mathieu and Prettenhofer, Peter and Weiss, Ron and Dubourg, Vincent and Vanderplas, Jake and Passos, Alexandre and Cournapeau, David and Brucher, Matthieu and Perrot, Matthieu and Duchesnay, Edouard},
	year = {2011},
	pages = {2825--2830},
}

@phdthesis{louppe_understanding_2014,
	title = {Understanding random forests: from theory to practice},
	shorttitle = {Understanding {Random} {Forests}},
	urldate = {2025-11-28},
	school = {University of Liège},
	author = {Louppe, Gilles},
	month = oct,
	year = {2014},
}

@article{boelts_simulation-based_2023,
	title = {Simulation-based inference for efficient identification of generative models in computational connectomics},
	volume = {19},
	issn = {1553-7358},
	doi = {10.1371/journal.pcbi.1011406},
	language = {en},
	number = {9},
	urldate = {2025-11-09},
	journal = {PLOS Computational Biology},
	author = {Boelts, Jan and Harth, Philipp and Gao, Richard and Udvary, Daniel and Yáñez, Felipe and Baum, Daniel and Hege, Hans-Christian and Oberlaender, Marcel and Macke, Jakob H.},
	month = sep,
	year = {2023},
	pages = {e1011406},
}

@article{hull_simulation-based_2024,
	title = {Simulation-based inference for parameter estimation of complex watershed simulators},
	volume = {28},
	issn = {1027-5606},
	doi = {10.5194/hess-28-4685-2024},
	language = {English},
	number = {20},
	urldate = {2025-11-09},
	journal = {Hydrology and Earth System Sciences},
	author = {Hull, Robert and Leonarduzzi, Elena and De La Fuente, Luis and Viet Tran, Hoang and Bennett, Andrew and Melchior, Peter and Maxwell, Reed M. and Condon, Laura E.},
	month = oct,
	year = {2024},
	pages = {4685--4713},
}

@article{greig_exploring_2024,
	title = {Exploring the role of the halo-mass function for inferring astrophysical parameters during reionization},
	volume = {533},
	issn = {0035-8711},
	doi = {10.1093/mnras/stae1983},
	number = {2},
	urldate = {2025-11-09},
	journal = {Monthly Notices of the Royal Astronomical Society},
	author = {Greig, Bradley and Prelogović, David and Mirocha, Jordan and Qin, Yuxiang and Ting, Yuan-Sen and Mesinger, Andrei},
	month = sep,
	year = {2024},
	pages = {2502--2529},
}

@book{breiman_classification_1984,
	address = {Belmont, Calif.},
	series = {Wadsworth statistics/probability series},
	title = {Classification and regression trees},
	isbn = {978-0-534-98053-5},
	urldate = {2025-11-06},
	publisher = {Wadsworth International Group},
	author = {Breiman, Leo},
	year = {1984},
}

@article{breiman_bagging_1996,
	title = {Bagging predictors},
	volume = {24},
	issn = {1573-0565},
	doi = {10.1007/BF00058655},
	language = {en},
	number = {2},
	urldate = {2025-11-06},
	journal = {Machine Learning},
	author = {Breiman, Leo},
	month = aug,
	year = {1996},
	pages = {123--140},
}

@article{meinshausen_quantile_2006,
	title = {Quantile regression forests},
	volume = {7},
	issn = {1533-7928},
	number = {35},
	urldate = {2025-10-28},
	journal = {Journal of Machine Learning Research},
	author = {Meinshausen, Nicolai},
	year = {2006},
	pages = {983--999},
}

@article{pudlo_reliable_2016,
	title = {Reliable {ABC} model choice via random forests},
	volume = {32},
	issn = {1367-4803},
	doi = {10.1093/bioinformatics/btv684},
	number = {6},
	urldate = {2025-10-26},
	journal = {Bioinformatics},
	author = {Pudlo, Pierre and Marin, Jean-Michel and Estoup, Arnaud and Cornuet, Jean-Marie and Gautier, Mathieu and Robert, Christian P.},
	month = mar,
	year = {2016},
	pages = {859--866},
}

@article{raynal_abc_2019,
	title = {{ABC} random forests for {Bayesian} parameter inference},
	volume = {35},
	issn = {1367-4803},
	doi = {10.1093/bioinformatics/bty867},
	number = {10},
	urldate = {2025-10-26},
	journal = {Bioinformatics},
	author = {Raynal, Louis and Marin, Jean-Michel and Pudlo, Pierre and Ribatet, Mathieu and Robert, Christian P and Estoup, Arnaud},
	month = may,
	year = {2019},
	pages = {1720--1728},
}

@article{de_santi_field-level_2025,
	title = {Field-level simulation-based inference with galaxy catalogs: the impact of systematic effects},
	volume = {2025},
	issn = {1475-7516},
	shorttitle = {Field-level simulation-based inference with galaxy catalogs},
	doi = {10.1088/1475-7516/2025/01/082},
	language = {en},
	number = {01},
	urldate = {2025-10-26},
	journal = {Journal of Cosmology and Astroparticle Physics},
	author = {de Santi, Natalí S.M. and Villaescusa-Navarro, Francisco and Raul Abramo, L. and Shao, Helen and Perez, Lucia A. and Castro, Tiago and Ni, Yueying and Lovell, Christopher C. and Hernández-Martínez, Elena and Marinacci, Federico and Spergel, David N. and Dolag, Klaus and Hernquist, Lars and Vogelsberger, Mark},
	month = jan,
	year = {2025},
	pages = {082},
}

@article{shih_fast_2024,
	title = {Fast parameter inference on pulsar timing arrays with normalizing flows},
	volume = {133},
	doi = {10.1103/PhysRevLett.133.011402},
	number = {1},
	journal = {Physical Review Letters},
	author = {Shih, David},
	year = {2024},
}

@article{dinh_approximate_2025,
	title = {Approximate {Bayesian} computation sequential {Monte} {Carlo} via random forests},
	volume = {35},
	issn = {1573-1375},
	doi = {10.1007/s11222-025-10748-x},
	language = {en},
	number = {6},
	urldate = {2025-10-26},
	journal = {Statistics and Computing},
	author = {Dinh, Khanh N. and Liu, Cécile and Xiang, Zijin and Liu, Zhihan and Tavaré, Simon},
	month = oct,
	year = {2025},
	pages = {219},
}

@article{hoffman_no-u-turn_2014,
	title = {The no-{U}-turn sampler: adaptively setting path lengths in {Hamiltonian} {Monte} {Carlo}},
	volume = {15},
	issn = {1533-7928},
	shorttitle = {The {No}-{U}-{Turn} {Sampler}},
	number = {47},
	urldate = {2025-09-07},
	journal = {Journal of Machine Learning Research},
	author = {Hoffman, Matthew D. and Gelman, Andrew},
	year = {2014},
	pages = {1593--1623},
}

@inproceedings{ward_robust_2022,
	title = {Robust neural posterior estimation and statistical model criticism},
	volume = {35},
	language = {en},
	urldate = {2024-07-05},
	booktitle = {Advances in {Neural} {Information} {Processing} {Systems}},
	author = {Ward, Daniel and Cannon, Patrick and Beaumont, Mark and Fasiolo, Matteo and Schmon, Sebastian M.},
	month = oct,
	year = {2022},
}

@inproceedings{schmitt_detecting_2024,
	address = {Cham},
	title = {Detecting model misspecification in amortized {Bayesian} inference with neural networks},
	isbn = {978-3-031-54605-1},
	doi = {10.1007/978-3-031-54605-1_35},
	language = {en},
	booktitle = {Pattern {Recognition}},
	publisher = {Springer Nature Switzerland},
	author = {Schmitt, Marvin and Bürkner, Paul-Christian and Köthe, Ullrich and Radev, Stefan T.},
	editor = {Köthe, Ullrich and Rother, Carsten},
	year = {2024},
	pages = {541--557},
}

@misc{wehenkel_addressing_2025,
	title = {Addressing misspecification in simulation-based inference through data-driven calibration},
	doi = {10.48550/arXiv.2405.08719},
	urldate = {2025-08-11},
	publisher = {arXiv},
	author = {Wehenkel, Antoine and Gamella, Juan L. and Sener, Ozan and Behrmann, Jens and Sapiro, Guillermo and Jacobsen, Jörn-Henrik and Cuturi, Marco},
	month = may,
	year = {2025},
	note = {arXiv preprint arXiv:2405.08719},
}

@article{zammit-mangion_neural_2025,
	title = {Neural methods for amortized inference},
	volume = {12},
	issn = {2326-8298, 2326-831X},
	doi = {10.1146/annurev-statistics-112723-034123},
	language = {en},
	urldate = {2025-08-11},
	journal = {Annual Review of Statistics and Its Application},
	author = {Zammit-Mangion, Andrew and Sainsbury-Dale, Matthew and Huser, Raphaël},
	month = mar,
	year = {2025},
	pages = {311--335},
}

@misc{kelly_simulation-based_2025,
	title = {Simulation-based {Bayesian} inference under model misspecification},
	doi = {10.48550/arXiv.2503.12315},
	urldate = {2025-06-01},
	publisher = {arXiv},
	author = {Kelly, Ryan P. and Warne, David J. and Frazier, David T. and Nott, David J. and Gutmann, Michael U. and Drovandi, Christopher},
	month = mar,
	year = {2025},
	note = {arXiv preprint arXiv:2503.12315},
}

@misc{mishra_robust_2025,
	title = {Robust amortized {Bayesian} inference with self-consistency losses on unlabeled data},
	doi = {10.48550/arXiv.2501.13483},
	urldate = {2025-05-11},
	publisher = {arXiv},
	author = {Mishra, Aayush and Habermann, Daniel and Schmitt, Marvin and Radev, Stefan T. and Bürkner, Paul-Christian},
	month = mar,
	year = {2025},
	note = {arXiv preprint arXiv:2501.13483},
}

@article{kleijn_bernstein-von-mises_2012,
	title = {The {Bernstein}-von-{Mises} theorem under misspecification},
	volume = {6},
	issn = {1935-7524, 1935-7524},
	doi = {10.1214/12-EJS675},
	urldate = {2023-07-04},
	journal = {Electronic Journal of Statistics},
	author = {Kleijn, B. J. K. and Vaart, A. W. van der},
	month = jan,
	year = {2012},
	pages = {354--381},
}

@inproceedings{huang_learning_2023,
	title = {Learning robust statistics for simulation-based inference under model misspecification},
	volume = {36},
	language = {en},
	urldate = {2024-07-05},
	booktitle = {Advances in {Neural} {Information} {Processing} {Systems}},
	author = {Huang, Daolang and Bharti, Ayush and Souza, Amauri and Acerbi, Luigi and Kaski, Samuel},
	month = dec,
	year = {2023},
}

@inproceedings{papamakarios_fast_2016,
	title = {Fast $\varepsilon$-free inference of simulation models with {Bayesian} conditional density estimation},
	volume = {29},
	urldate = {2022-07-29},
	booktitle = {Advances in {Neural} {Information} {Processing} {Systems}},
	author = {Papamakarios, George and Murray, Iain},
	year = {2016},
}

@misc{cranmer_approximating_2016,
	title = {Approximating likelihood ratios with calibrated discriminative classifiers},
	doi = {10.48550/arXiv.1506.02169},
	urldate = {2025-02-09},
	publisher = {arXiv},
	author = {Cranmer, Kyle and Pavez, Juan and Louppe, Gilles},
	month = mar,
	year = {2016},
	note = {arXiv preprint arXiv:1506.02169},
}

@article{marin_relevant_2014,
	title = {Relevant statistics for {Bayesian} model choice},
	volume = {76},
	issn = {1467-9868},
	doi = {10.1111/rssb.12056},
	language = {en},
	number = {5},
	urldate = {2022-11-30},
	journal = {Journal of the Royal Statistical Society: Series B (Statistical Methodology)},
	author = {Marin, Jean-Michel and Pillai, Natesh S. and Robert, Christian P. and Rousseau, Judith},
	year = {2014},
	pages = {833--859},
}

@article{frazier_robust_2021,
	title = {Robust approximate {Bayesian} inference with synthetic likelihood},
	volume = {30},
	issn = {1061-8600},
	doi = {10.1080/10618600.2021.1875839},
	number = {4},
	urldate = {2022-07-30},
	journal = {Journal of Computational and Graphical Statistics},
	author = {Frazier, David T. and Drovandi, Christopher},
	month = oct,
	year = {2021},
	pages = {958--976},
}

@article{berk_limiting_1966,
	title = {Limiting behavior of posterior distributions when the model is incorrect},
	volume = {37},
	issn = {0003-4851, 2168-8990},
	doi = {10.1214/aoms/1177699597},
	number = {1},
	urldate = {2024-06-02},
	journal = {The Annals of Mathematical Statistics},
	author = {Berk, Robert H.},
	month = feb,
	year = {1966},
	pages = {51--58},
}

@misc{frazier_statistical_2024,
	title = {The statistical accuracy of neural posterior and likelihood estimation},
	doi = {10.48550/arXiv.2411.12068},
	urldate = {2024-11-25},
	publisher = {arXiv},
	author = {Frazier, David T. and Kelly, Ryan and Drovandi, Christopher and Warne, David J.},
	month = nov,
	year = {2024},
	note = {arXiv preprint arXiv:2411.12068},
}

@article{wilkinson_approximate_2013,
	title = {Approximate {Bayesian} computation ({ABC}) gives exact results under the assumption of model error},
	volume = {12},
	issn = {1544-6115},
	doi = {10.1515/sagmb-2013-0010},
	language = {en},
	number = {2},
	urldate = {2022-09-14},
	journal = {Statistical Applications in Genetics and Molecular Biology},
	author = {Wilkinson, Richard David},
	month = may,
	year = {2013},
	pages = {129--141},
}

@article{ratmann_model_2009,
	title = {Model criticism based on likelihood-free inference, with an application to protein network evolution},
	volume = {106},
	doi = {10.1073/pnas.0807882106},
	number = {26},
	urldate = {2024-05-01},
	journal = {Proceedings of the National Academy of Sciences},
	author = {Ratmann, Oliver and Andrieu, Christophe and Wiuf, Carsten and Richardson, Sylvia},
	month = jun,
	year = {2009},
	pages = {10576--10581},
}

@article{radev_bayesflow_2023,
	title = {{BayesFlow}: amortized {Bayesian} workflows with neural networks},
	volume = {8},
	issn = {2475-9066},
	shorttitle = {{BayesFlow}},
	doi = {10.21105/joss.05702},
	language = {en},
	number = {89},
	urldate = {2024-07-05},
	journal = {Journal of Open Source Software},
	author = {Radev, Stefan T. and Schmitt, Marvin and Schumacher, Lukas and Elsemüller, Lasse and Pratz, Valentin and Schälte, Yannik and Köthe, Ullrich and Bürkner, Paul-Christian},
	month = sep,
	year = {2023},
	pages = {5702},
}

@article{kennedy_bayesian_2001,
	title = {Bayesian calibration of computer models},
	volume = {63},
	issn = {1467-9868},
	doi = {10.1111/1467-9868.00294},
	language = {en},
	number = {3},
	urldate = {2024-03-26},
	journal = {Journal of the Royal Statistical Society: Series B (Statistical Methodology)},
	author = {Kennedy, Marc C. and O'Hagan, Anthony},
	year = {2001},
	pages = {425--464},
}

@misc{cannon_investigating_2022,
	title = {Investigating the impact of model misspecification in neural simulation-based inference},
	doi = {10.48550/arXiv.2209.01845},
	urldate = {2022-09-08},
	publisher = {arXiv},
	author = {Cannon, Patrick and Ward, Daniel and Schmon, Sebastian M.},
	month = sep,
	year = {2022},
	note = {arXiv preprint arXiv:2209.01845},
}

@article{wood_statistical_2010,
	title = {Statistical inference for noisy nonlinear ecological dynamic systems},
	volume = {466},
	copyright = {2010 Springer Nature Limited},
	issn = {1476-4687},
	doi = {10.1038/nature09319},
	language = {en},
	number = {7310},
	urldate = {2024-05-01},
	journal = {Nature},
	author = {Wood, Simon N.},
	month = aug,
	year = {2010},
	pages = {1102--1104},
}

@article{sisson_sequential_2007,
	title = {Sequential {Monte} {Carlo} without likelihoods},
	volume = {104},
	doi = {10.1073/pnas.0607208104},
	number = {6},
	urldate = {2024-05-08},
	journal = {Proceedings of the National Academy of Sciences},
	author = {Sisson, S. A. and Fan, Y. and Tanaka, Mark M.},
	month = feb,
	year = {2007},
	pages = {1760--1765},
}

@article{price_bayesian_2018,
	title = {Bayesian synthetic likelihood},
	volume = {27},
	issn = {1061-8600},
	doi = {10.1080/10618600.2017.1302882},
	number = {1},
	urldate = {2024-05-01},
	journal = {Journal of Computational and Graphical Statistics},
	author = {Price, L. F. and Drovandi, C. C. and Lee, A. and Nott, D. J.},
	month = jan,
	year = {2018},
	pages = {1--11},
}

@article{marjoram_markov_2003,
	title = {Markov chain {Monte} {Carlo} without likelihoods},
	volume = {100},
	doi = {10.1073/pnas.0306899100},
	number = {26},
	urldate = {2024-05-08},
	journal = {Proceedings of the National Academy of Sciences},
	author = {Marjoram, Paul and Molitor, John and Plagnol, Vincent and Tavaré, Simon},
	month = dec,
	year = {2003},
	pages = {15324--15328},
}

@article{frazier_model_2020,
	title = {Model misspecification in approximate {Bayesian} computation: consequences and diagnostics},
	volume = {82},
	issn = {1369-7412},
	shorttitle = {Model {Misspecification} in {Approximate} {Bayesian} {Computation}},
	doi = {10.1111/rssb.12356},
	number = {2},
	urldate = {2024-05-08},
	journal = {Journal of the Royal Statistical Society Series B: Statistical Methodology},
	author = {Frazier, David T. and Robert, Christian P. and Rousseau, Judith},
	month = apr,
	year = {2020},
	pages = {421--444},
}

@article{csillery_abc_2012,
	title = {abc: an {R} package for approximate {Bayesian} computation ({ABC})},
	volume = {3},
	issn = {2041-210X},
	shorttitle = {abc},
	doi = {10.1111/j.2041-210X.2011.00179.x},
	language = {en},
	number = {3},
	urldate = {2023-11-22},
	journal = {Methods in Ecology and Evolution},
	author = {Csilléry, Katalin and François, Olivier and Blum, Michael G. B.},
	year = {2012},
	pages = {475--479},
}

@article{cranmer_frontier_2020,
	title = {The frontier of simulation-based inference},
	volume = {117},
	doi = {10.1073/pnas.1912789117},
	number = {48},
	urldate = {2024-05-27},
	journal = {Proceedings of the National Academy of Sciences},
	author = {Cranmer, Kyle and Brehmer, Johann and Louppe, Gilles},
	month = dec,
	year = {2020},
	pages = {30055--30062},
}

@article{blum_approximate_2010,
	title = {Approximate {Bayesian} computation: a nonparametric perspective},
	volume = {105},
	issn = {0162-1459},
	shorttitle = {Approximate {Bayesian} {Computation}},
	doi = {10.1198/jasa.2010.tm09448},
	number = {491},
	urldate = {2024-06-27},
	journal = {Journal of the American Statistical Association},
	author = {Blum, Michael G. B.},
	month = sep,
	year = {2010},
	pages = {1178--1187},
}

@article{beaumont_approximate_2002,
	title = {Approximate {Bayesian} computation in population genetics},
	volume = {162},
	issn = {1943-2631},
	doi = {10.1093/genetics/162.4.2025},
	number = {4},
	urldate = {2024-06-19},
	journal = {Genetics},
	author = {Beaumont, Mark A and Zhang, Wenyang and Balding, David J},
	month = dec,
	year = {2002},
	pages = {2025--2035},
}

@article{bayarri_modularization_2009,
	title = {Modularization in {Bayesian} analysis, with emphasis on analysis of computer models},
	volume = {4},
	issn = {1936-0975, 1931-6690},
	doi = {10.1214/09-BA404},
	number = {1},
	urldate = {2024-03-26},
	journal = {Bayesian Analysis},
	author = {Bayarri, M. J. and Berger, J. O. and Liu, F.},
	month = mar,
	year = {2009},
	pages = {119--150},
}

@article{barber_rate_2015,
	title = {The rate of convergence for approximate {Bayesian} computation},
	volume = {9},
	issn = {1935-7524, 1935-7524},
	doi = {10.1214/15-EJS988},
	number = {1},
	urldate = {2024-05-08},
	journal = {Electronic Journal of Statistics},
	author = {Barber, Stuart and Voss, Jochen and Webster, Mark},
	month = jan,
	year = {2015},
	pages = {80--105},
}

@article{walker_bayesian_2013,
	title = {Bayesian inference with misspecified models},
	volume = {143},
	issn = {0378-3758},
	doi = {10.1016/j.jspi.2013.05.013},
	language = {en},
	number = {10},
	urldate = {2023-06-05},
	journal = {Journal of Statistical Planning and Inference},
	author = {Walker, Stephen G.},
	month = oct,
	year = {2013},
	pages = {1621--1633},
}
\appendix

\section{SMC ABC background}\label{sec:app_smcabc_background}
Sequential Monte Carlo approximate Bayesian computation (SMC ABC) targets a sequence of joint distributions
\begin{equation*}
    \pi_t(\bm{\theta}, \bm{x} \mid S(\bm{y}), \epsilon_t) \propto \mathbb{I}\{\rho(S(\bm{x}), S(\bm{y})) \leq \epsilon_t \} p(\bm{x} \mid \bm{\theta}) p(\bm{\theta}),
\end{equation*}
defined by a decreasing tolerance schedule $\epsilon_1 \geq \cdots \geq \epsilon_T$.
While Rejection ABC produces independent but inefficient draws \citep{beaumont_approximate_2002}, MCMC ABC improves efficiency at the cost of Markov dependence and tuning \citep{marjoram_markov_2003}.
We adopt the replenishment SMC ABC algorithm of \citet{drovandi_estimation_2011}, which propagates particles by resampling and applying an MCMC mutation kernel invariant to $\pi_t$.

At iteration $t$, let $\{ (\bm{\theta}_i, \bm{s}_i, \rho_i) \}_{i=1}^N$ denote the population of particles, summaries, and discrepancies.
We sort the population by $\rho_i$ and discard the fraction $\alpha \in (0, 1)$ with the highest discrepancies. We retain $N_{\text{alive}} = N - \lfloor \alpha N \rfloor$ particles and set the next tolerance to the largest retained distance,
\begin{equation*}
    \epsilon_{t+1} = \max_{i \leq N_{\text{alive}}} \rho_i.
\end{equation*}
We replenish the population by resampling $\lfloor \alpha N \rfloor$ particles from the alive set (with replacement), then applying an MCMC ABC move to each duplicate using a symmetric random-walk proposal
\begin{equation*}
    \bm{\theta}^{\prime} \sim \mathcal{N}(\bm{\theta}, \Sigma_t),
    \end{equation*}
where $\Sigma_t$ is the sample covariance of the alive parameters. We simulate $\bm{x}^{\prime} \sim P_{\bm{\theta}^{\prime}}$, compute $\bm{s}' = S(\bm{x}^{\prime})$ and $\rho^{\prime} = \rho(\bm{s}', \bm{s}_y)$, and accept the move with probability
\begin{equation*}
  \alpha_{\text{acc}} = \min \left(1, \frac{\pi(\bm{\theta}^{\prime})}{\pi(\bm{\theta})} \mathbb{I}\{\rho^{\prime} \leq \epsilon_{t+1}\} \right).
\end{equation*}

The number of MCMC updates $R_t$ applied to each resampled particle is adapted to ensure diversity. Let $\hat{p}_t$ denote the estimated acceptance probability for the MCMC kernel. We set
\begin{equation*}
    R_t = \left\lceil \frac{\log c}{\log{1-\hat{p}_t}} \right\rceil,
\end{equation*}
where $c \in (0, 1)$ is the target probability that a resampled particle remains a duplicate after $R_t$ trials (typically $c=0.01$).
The algorithm terminates when $\epsilon_{t+1} \leq \epsilon_{\min}$, the acceptance probability drops below a threshold $p_{\min}$, or a maximum number of generations is reached.
This approach avoids the need for a pre-specified tolerance schedule, adaptively determining intermediate targets while maintaining particle diversity \citep{drovandi_estimation_2011}.

\begin{algorithm}[htbp]
\caption{Adaptive Replenishment SMC ABC \citep{drovandi_estimation_2011}}
\label{alg:smcabc-replenish}
\begin{algorithmic}[1]
\Require Observed summary $\bm{s}_y$; prior $\pi(\bm{\theta})$; simulator $P_{\bm{\theta}}$; discrepancy $\rho$; population size $N$; drop fraction $\alpha$; min tolerance $\epsilon_{\min}$; min acceptance $p_{\min}$; duplication prob $c$; max generations $T_{\max}$.
\Ensure Posterior samples $\{\bm{\theta}_i\}_{i=1}^N$ (and associated summaries).
\State \textbf{Initialise:} For $i=1,\dots,N$: draw $\bm{\theta}_i \sim \pi(\bm{\theta})$, simulate $\bm{x}_i \sim P_{\bm{\theta}_i}$, compute $\bm{s}_i \gets S(\bm{x}_i)$ and $\rho_i \gets \rho(\bm{s}_i, \bm{s}_y)$.
\State Set $\epsilon \gets \max_i \rho_i$, $R \gets 1$, $t \gets 1$.
\While{$t \le T_{\max}$}
  \State Sort population by $\rho_i$ ascending.
  \State Determine cutoff index $N_{\text{alive}} \gets N - \lfloor \alpha N \rfloor$.
  \State Update tolerance $\epsilon \gets \rho_{N_{\text{alive}}}$ (discrepancy of the worst kept particle).
  \State Compute covariance $\Sigma_t$ of the $N_{\text{alive}}$ kept parameters.
  \State Resample $\lfloor \alpha N \rfloor$ indices from the alive set to form the replenishment set $\mathcal{J}$.
  \State $N_{\text{acc}} \gets 0$
  \For{$j \in \mathcal{J}$} \Comment{Apply MCMC diversity steps}
    \For{$r=1,\dots,R$}
      \State Propose $\bm{\theta}' \sim \mathcal{N}(\bm{\theta}_j, \Sigma_t)$.
      \State Simulate $\bm{x}' \sim P_{\bm{\theta}'}$, compute $\bm{s}' \gets S(\bm{x}')$ and $\rho' \gets \rho(\bm{s}', \bm{s}_y)$.
      \State Compute $a \gets \min\left(1, \frac{\pi(\bm{\theta}')}{\pi(\bm{\theta}_j)} \mathbb{I}\{\rho' \le \epsilon\} \right)$.
      \If{$\mathcal{U}(0,1) < a$}
        \State $\bm{\theta}_j \gets \bm{\theta}'$, \quad $\bm{s}_j \gets \bm{s}'$, \quad $\rho_j \gets \rho'$, \quad $N_{\text{acc}} \gets N_{\text{acc}} + 1$.
      \EndIf
    \EndFor
  \EndFor
  \State Update $\hat{p}_t \gets N_{\text{acc}} / (\lfloor \alpha N \rfloor \times R)$.
  \State Update $R \gets \max\left(1, \lceil \frac{\log c}{\log(1-\hat{p}_t)} \rceil \right)$.
  \State Reassemble population (Alive + Replenished).
  \If{$\epsilon \le \epsilon_{\min}$ \textbf{or} $\hat{p}_t < p_{\min}$} 
    \State \textbf{break} 
  \EndIf
  \State $t \gets t+1$.
\EndWhile
\State \Return $\{\bm{\theta}_i\}_{i=1}^N$.
\end{algorithmic}
\end{algorithm}


\section{Random forest background}\label{sec:app_rf_background}
We review regression trees and random forests, focusing on the leaf-based weights used to define tolerance-free, data-adaptive neighbourhoods around the observed summary.

Regression trees (CART; \citet{breiman_classification_1984}) recursively partition the summary space $\mathcal{S} \subset \mathbb{R}^{d_s}$ via axis‑aligned binary splits that greedily minimise within‑node squared error. This yields a piecewise‑constant predictor equal to the mean response in each terminal region. At a node with index set $\mathcal{I}$, we select the feature $j\in\{1,\dots,d_s\}$ and threshold $\tau$ that maximise the impurity decrease
\begin{equation*}    
\Delta = |\mathcal{I}|\operatorname{Var}_{\mathcal{I}}(t) - |\mathcal{I}_L|\operatorname{Var}_L(t) - |\mathcal{I}_R|\operatorname{Var}_R(t),
\end{equation*}
where the split induces the partition\begin{equation*}
\mathcal{I}_L=\{i\in \mathcal{I}  \colon s_{ij}\le \tau\}, \quad \mathcal{I}_R= \mathcal{I} \setminus \mathcal{I}_L.
\end{equation*}
Recursion halts when depth reaches $D_{\max}$ or the node size falls below $m_{\min}$; such pre‑pruning controls prevent overfitting and limit computational cost. For any terminal region $R_\ell$ with index set $I_\ell$, the predictor is the local mean
\begin{equation*}
\hat f(\bm{s})=\bar t_\ell=\frac{1}{|I_\ell|}\sum_{i\in I_\ell} t_i, \quad \forall \bm{s}\in R_\ell,
\end{equation*}
which minimises the squared error within the leaf.
While single trees have high variance, averaging many trees via bagging reduces variance \citep{breiman_bagging_1996}; adding random feature selection at each split yields random forests \citep{breiman_random_2001}.

A random forest aggregates $B$ randomised regression trees to reduce variance and define proximities in $\mathcal{S}$ via leaf co-membership \citep{breiman_random_2001}.
For each tree $b$, we draw a bootstrap sample $D_b$ from the training set $\mathcal{D}=\{(\bm{s}_i,t_i)\}_{i=1}^N$; this resampling adds variability and enables out‑of‑bag error assessment.
At each split, we consider a random subset of $m_{\text{try}}$ features to decorrelate the trees and improve ensemble diversity.
Predictions are aggregated via averaging, $\hat f^{\mathrm{RF}}(\bm{s})=\frac{1}{B}\sum_{b=1}^B \hat f_b(\bm{s})$.
Beyond point prediction, the forest structure defines a kernel based on the leaf $L_b(\bm{s})$ containing a query $\bm{s}$. The proximity between two points is the proportion of trees in which they share a leaf:
\begin{equation*}
    \mathrm{prox}(\bm{s},\bm{s}')=\frac{1}{B}\sum_{b=1}^{B}\mathbb{I}\{L_b(\bm{s})=L_b(\bm{s}')\}.
\end{equation*}
To construct data-adaptive weights for preconditioning, we utilise the per-leaf uniform weights defined by \citet{meinshausen_quantile_2006}:
\begin{equation*}
    v_i^{(b)}(\bm{s}) = \frac{\mathbb{I}\{\bm{s}_i\in L_b(\bm{s})\}}{|L_b(\bm{s})|},
\end{equation*}
which distribute mass uniformly among the training samples co-occurring with $\bm{s}$ in tree $b$. Algorithm~\ref{alg:train-rf} summarises the standard training loop.

\begin{algorithm}[htbp]
\caption{Random forest training}
\label{alg:train-rf}
\begin{algorithmic}[1]
\Require Dataset $\mathcal{D}=\{(\bm{s}_i,\theta_i)\}_{i=1}^N$; number of trees $B$; Bootstrap flag.
\Require Tree topology constraints: $n_{\min}$ (leaf size), $h_{\max}$ (depth), $m_{\text{try}}$ (features), $\delta$ (impurity).
\Ensure Ensemble $\mathcal{F}=\{T_b\}_{b=1}^{B}$ and out-of-bag sets $\{\mathcal{O}_b\}_{b=1}^{B}$.
\State Initialise $\mathcal{F} \gets \emptyset$, $\mathcal{O} \gets \emptyset$.
\For{$b=1, \dots, B$}
  \If{Bootstrap}
    \State Draw $N$ indices $I_b$ from $\{1,\dots,N\}$ with replacement.
    \State Record out-of-bag indices $\mathcal{O}_b \gets \{1,\dots,N\} \setminus I_b$.
  \Else
    \State Set $I_b \gets \{1,\dots,N\}$ and $\mathcal{O}_b \gets \emptyset$.
  \EndIf
  \State Construct subset $\mathcal{D}_b \gets \{(\bm{s}_i, \theta_i) : i \in I_b\}$.
  \State Grow regression tree $T_b$ on $\mathcal{D}_b$ subject to topology constraints.
  \State Add $T_b$ to $\mathcal{F}$ and $\mathcal{O}_b$ to $\mathcal{O}$.
\EndFor
\State \Return $\mathcal{F}, \mathcal{O}$.
\end{algorithmic}
\end{algorithm}

In the context of ABC, forest proximities define a probability mass function (pmf) over the training simulations, centred at the observed summary $\bm{s}_y$ \citep{raynal_abc_2019}.
Let $L_b(\bm{s}_y)$ denote the leaf in tree $b$ containing $\bm{s}_y$. To account for bootstrap resampling, let $n_{ib}$ be the number of times simulation $i$ appears in the bootstrap sample used to train tree $b$.
The effective size of the leaf is the total count of in-bag samples it contains: $|L_b(\bm{s}_y)| = \sum_{i=1}^N n_{ib} \mathbb{I}\{\bm{s}_i \in L_b(\bm{s}_y)\}$.
Following \citet{meinshausen_quantile_2006}, we define the per-tree weight for simulation $i$ as
\begin{equation*}
    w_{ib}(\bm{s}_y) = \frac{n_{ib} \mathbb{I}\{\bm{s}_i \in L_b(\bm{s}_y)\}}{|L_b(\bm{s}_y)|}.
\end{equation*}
Aggregating across the forest yields the final weight $W_i^{\mathrm{RF}}(\bm{s}_y) = \frac{1}{B} \sum_{b=1}^B w_{ib}(\bm{s}_y)$, which allows for the estimation of conditional expectations via $\widetilde{\mathbb{E}}[\theta \mid \bm{s}_y] = \sum_{i=1}^N W_i^{\mathrm{RF}}(\bm{s}_y) \theta_i$.
Unlike the standard proximity matrix $P_{ij}$, which is symmetric and whose rows do not sum to one, the leaf-normalised weights $W^{\mathrm{RF}}$ constitute a valid distribution suitable for importance sampling and ESS calculation.

In practice, we simplify the weighting scheme by ignoring bootstrap multiplicities. For a given tree $b$, the weight becomes
\begin{equation*}
    \bar{w}_{ib}(\bm{s}_y) = \frac{\mathbb{I}\{\bm{s}_i \in L_b(\bm{s}_y)\}}{\sum_{k=1}^N \mathbb{I}\{\bm{s}_k \in L_b(\bm{s}_y)\}}.
\end{equation*}
For large $N$, this approximation closely matches the multiplicity-aware weights while reducing implementation complexity \citep{meinshausen_quantile_2006}.
We aggregate these weights across the $B$ trees of all $d_{\bm{\theta}}$ per-parameter forests. The final weight for simulation $i$ is
\begin{equation*}
    W^{\mathrm{RF}}(\bm{s}_i;\bm{s}_y) = \frac{1}{d_{\bm{\theta}}B} \sum_{j=1}^{d_{\bm{\theta}}} \sum_{b=1}^{B} \frac{\mathbb{I}\{\bm{s}_i \in L_{jb}(\bm{s}_y)\}}{|L_{jb}(\bm{s}_y)|}.
\end{equation*}
This summation yields a valid probability mass function ($\sum_i W^{\mathrm{RF}}_i = 1$) suitable for importance sampling \citep{raynal_abc_2019}.
While multi-output forests are possible, we employ per-parameter forests to allow the splitting rules to adapt specifically to the sensitivity of each parameter component $\theta_j$.
Finally, we quantify the concentration of the measure using the effective sample size (ESS), $N_{\mathrm{eff}} = (\sum_{i=1}^{N}\tilde w_i^{2})^{-1}$, which guides the decision to resample.

The computational overhead of the random forest step is negligible relative to the cost of the simulator. Training complexity scales quasilinearly with the sample size $N$ and linearly with the total number of trees $B_{\text{tot}}$ \citep{louppe_understanding_2014}. Furthermore, querying the weights at $\bm{s}_y$ requires only a single path traversal per tree, which is computationally instantaneous. Given that the simulator typically dominates the wall-time budget, and tree-based ensembles scale efficiently to large datasets \citep{breiman_random_2001, louppe_understanding_2014}, the forest preconditioning step does not introduce a bottleneck.

\section{Hyperparameter setup}\label{sec:hyperparameters}
\textbf{Normalising flow}:
We use a conditional neural spline flow \citep{durkan_neural_2019} implemented in \textsf{flowjax} \citep{ward_flowjax_2025}. The flow comprises eight coupling layers, each with a rational quadratic spline transformer using 10 knots over the interval $[-8, 8]$, outside this range the transformer defaults to the identity, which is important when the observed summaries lie in the tails under misspecification. Each coupling layer's conditioner is a multilayer perceptron with 128 hidden units. Both summary statistics and parameters are standardised to zero mean and unit variance before training. When the prior has compact support, parameters are first mapped to $\mathbb{R}$ via the logit transform and standardised in unconstrained space, posterior samples are mapped back through the sigmoid. The flow is trained by minimising the negative log-likelihood with the Adam optimiser \citep{kingma_adam_2017} at a learning rate of $5\times10^{-4}$ and a batch size of 512. Training is stopped when the validation loss has not improved for 10 consecutive epochs, or after 500 epochs, whichever occurs first.

\textbf{RNPE denoiser}:
We use the spike-and-slab error model of \citet{ward_robust_2022} with a Gaussian spike $\mathcal{N}(0, \sigma_{\mathrm{spike}}^2)$ and Cauchy slab $\mathrm{Cauchy}(0, \sigma_{\mathrm{slab}})$, setting $\sigma_{\mathrm{spike}} = 0.01$ and $\sigma_{\mathrm{slab}} = 0.25$.
The marginal summary density $h_{\bm{\psi}}(\bm{s})$ is an unconditional neural spline flow sharing the same architecture as the posterior flow (eight coupling layers, 10 knots over $[-8,8]$, single-hidden-layer MLP with 128 units), trained on the whitened, preconditioned summaries.

\textbf{MCMC (denoising sampler}):
Denoised summaries are drawn via a single NUTS chain \citep{hoffman_no-u-turn_2014} implemented in \textsf{NumPyro} \citep{phan_composable_2019}, with target acceptance probability $0.9$, 1,000 warmup iterations, and 2,000 post-warmup samples with no thinning.

\textbf{SMC ABC}:
We use the adaptive replenishment SMC-ABC algorithm of \citet{drovandi_estimation_2011}, as described in Appendix~\ref{sec:app_smcabc_background}. The population size $N$ is set per example, typically $N=4,000$. We initialise with a large tolerance $\varepsilon_0=10^{6}$ and at each generation reduce it to the $\alpha$-quantile of the current distance distribution with $\alpha=0.5$, subject to a floor of $\varepsilon_{\min}=10^{-3}$. Resampled particles are diversified via the MCMC kernel of Algorithm~\ref{alg:smcabc-replenish} with proposal covariance equal to the empirical covariance of the alive set and duplication probability $c=0.01$. The algorithm terminates when the acceptance rate falls below $0.10$ or after $T_{\max}=3$ generations. The initial number of MCMC repeats is $R=1$ with subsequent values adapted as described in Appendix~\ref{sec:app_smcabc_background}. We use Euclidean distance as the default discrepancy unless otherwise stated.

\textbf{Random forest}:
We implement forest-proximity preconditioning using \textsf{scikit-learn} \citep{pedregosa_scikit-learn_2011}.
By default, we train a separate random forest regressor for each parameter component $\theta_j$.
Each forest consists of 800 trees. To ensure conservative weighting, we constrain the topology with a maximum depth of 10 and a minimum leaf size of 40 samples (implying a minimum split size of 80).
We utilise bootstrap sampling with out-of-bag evaluation enabled.
In the default per-parameter setting, we consider all summary features at every split; for the multi-output alternative, we consider the square root of the feature count.
Splits require a minimum impurity decrease of $10^{-6}$.
Calculations are parallelised across all available processor cores.
Optionally, training can be restricted to a subset of size $m_{\text{fit}}=\lceil \rho N\rceil$ with fraction $\rho\in(0,1]$.
The final weights are derived from the leaf proximities as defined in Algorithm~2 in the main text.

\section{Proof of the amortisation-gap bound}\label{app:am_gap_proof}

We establish the bound on the amortisation gap used in Section~3.2 in the main text, accounting for summary-dependent preconditioning weights.

\paragraph{Setup.}
Recall the training design $p_{\text{train}}(\bm{\theta},\bm{s}) = \pi(\bm{\theta})\,p(\bm{s}\mid\bm{\theta})$ with marginal $p_{\text{train}}(\bm{s}) = \int p_{\text{train}}(\bm{\theta},\bm{s})\,\mathrm{d}\bm{\theta}$ and conditional $\pi(\bm{\theta}\mid\bm{s}) = p_{\text{train}}(\bm{\theta},\bm{s})/p_{\text{train}}(\bm{s})$.
For each $\bm{s}\in\mathcal{S}$, define the pointwise cross-entropy
\begin{equation}\label{eq:pointwise_loss}
\mathcal{L}(\bm{\phi};\bm{s}) = \int -\log q_{\bm{\phi}}(\bm{\theta}\mid\bm{s})\,\pi(\bm{\theta}\mid\bm{s})\,\mathrm{d}\bm{\theta}.
\end{equation}
Let $w_y\colon\mathcal{S}\to[0,\infty)$ be a data-dependent weight function.
Assume without loss of generality that $\mathbb{E}_{p_{\text{train}}}[w_y(\bm{s})] = 1$; since $\bm{\phi}^{\star}_w$ is invariant to rescaling of $w_y$, this is without loss.
Define the following minimisers:
\begin{align*}
\hat{\bm{\phi}}(\bm{s}_y) &:= \argmin_{\bm{\phi}\in\Phi}\;\mathcal{L}(\bm{\phi};\bm{s}_y), \\
\bm{\phi}^{\star}_w &:= \argmin_{\bm{\phi}\in\Phi}\;\int \mathcal{L}(\bm{\phi};\bm{s})\,w_y(\bm{s})\,p_{\text{train}}(\bm{s})\,\mathrm{d}\bm{s}.
\end{align*}
Here $\hat{\bm{\phi}}(\bm{s}_y)$ is the local minimiser at the observation, and $\bm{\phi}^{\star}_w$ is the population minimiser of the weighted global risk.
The amortisation gap under the weighted design is
\begin{equation}\label{eq:am_gap_w}
\Delta_{\mathrm{am}}(\bm{s}_y) := \mathcal{L}(\bm{\phi}^{\star}_w;\bm{s}_y) - \mathcal{L}(\hat{\bm{\phi}}(\bm{s}_y);\bm{s}_y) \ge 0,
\end{equation}
where positivity follows from the definition of $\hat{\bm{\phi}}(\bm{s}_y)$.
Setting $w_y \equiv 1$ recovers the standard (unweighted) amortisation gap.

\begin{assumption}\label{ass:lipschitz}
The following conditions hold.
\begin{enumerate}
    \item[(i)] There exist a positive function $C_1\colon\Theta\times\Phi\to\mathbb{R}_+$ such that, for any $\bm{s},\bm{s}'\in\mathcal{S}$,
    \[
    |\log q_{\bm{\phi}}(\bm{\theta}\mid\bm{s}) - \log q_{\bm{\phi}}(\bm{\theta}\mid\bm{s}')| \le C_1(\bm{\theta},\bm{\phi})\,\|\bm{s}-\bm{s}'\|.
    \]
    For $\bm{\phi}\in\{\bm{\phi}^{\star}_w,\,\hat{\bm{\phi}}(\bm{s}_y)\}$, define
    \[
    \bar{C}_1 := \int C_1(\bm{\theta},\bm{\phi})\,\pi(\bm{\theta}\mid\bm{s}_y)\,\mathrm{d}\bm{\theta} < \infty, \qquad
    \bar{C}_2 := \sqrt{\int C_1^2(\bm{\theta},\bm{\phi})\,w_y(\bm{s})\,p_{\text{train}}(\bm{\theta},\bm{s})\,\mathrm{d}\bm{\theta}\,\mathrm{d}\bm{s}} < \infty.
    \]

    \item[(ii)] There exist a positive function $f\colon\Theta\times\mathcal{S}\to\mathbb{R}_+$ and a constant $\kappa>1$ such that, for any $\bm{s}\in\mathcal{S}$,
    \[
    |\pi(\bm{\theta}\mid\bm{s}_y) - \pi(\bm{\theta}\mid\bm{s})| \le f(\bm{\theta}\mid\bm{s}_y)\,\|\bm{s}-\bm{s}_y\|^{\kappa},
    \]
    where
    \[
    \bar{C}_3 := \int f(\bm{\theta}\mid\bm{s}_y)\,|\log q_{\bm{\phi}}(\bm{\theta}\mid\bm{s}_y)|\,\mathrm{d}\bm{\theta} < \infty
    \]
    for $\bm{\phi}\in\{\bm{\phi}^{\star}_w,\,\hat{\bm{\phi}}(\bm{s}_y)\}$.
\end{enumerate}
\end{assumption}
We now restate Lemma 1 in the main text here and prove the result. 
\begin{lemma}[Amortisation-gap bound]\label{lem:am_gap_appendix}
Under Assumption~\ref{ass:lipschitz},
\begin{equation}\label{eq:am_gap_bound_full}
\begin{aligned}
\Delta_{\mathrm{am}}(\bm{s}_y)
&\le 4\bar{C}_1\,\mathbb{E}_{p_{\text{train}}(\bm{s})}
    \big[w_y(\bm{s})\|\bm{s}-\bm{s}_y\|\big] \\
&\quad + 2\bar{C}_2\sqrt{\mathbb{E}_{p_{\text{train}}(\bm{s})}
    \big[w_y(\bm{s})\|\bm{s}-\bm{s}_y\|^2\big]} \\
&\quad + 2\bar{C}_3\,\mathbb{E}_{p_{\text{train}}(\bm{s})}
    \big[w_y(\bm{s})\|\bm{s}-\bm{s}_y\|^{\kappa}\big].
\end{aligned}
\end{equation}
\end{lemma}

\begin{proof}
Decompose the amortisation gap as
\begin{align*}
\Delta_{\mathrm{am}}(\bm{s}_y) &= \mathcal{L}(\bm{\phi}^{\star}_w;\bm{s}_y) - \mathcal{L}(\hat{\bm{\phi}}(\bm{s}_y);\bm{s}_y) \\
&= \underbrace{\mathcal{L}(\bm{\phi}^{\star}_w;\bm{s}_y) - \int \mathcal{L}(\bm{\phi}^{\star}_w;\bm{s})\,w_y(\bm{s})\,p_{\text{train}}(\bm{s})\,\mathrm{d}\bm{s}}_{A} \\
&\quad + \underbrace{\int \mathcal{L}(\bm{\phi}^{\star}_w;\bm{s})\,w_y(\bm{s})\,p_{\text{train}}(\bm{s})\,\mathrm{d}\bm{s} - \int \mathcal{L}(\hat{\bm{\phi}}(\bm{s}_y);\bm{s})\,w_y(\bm{s})\,p_{\text{train}}(\bm{s})\,\mathrm{d}\bm{s}}_{B} \\
&\quad + \underbrace{\int \mathcal{L}(\hat{\bm{\phi}}(\bm{s}_y);\bm{s})\,w_y(\bm{s})\,p_{\text{train}}(\bm{s})\,\mathrm{d}\bm{s} - \mathcal{L}(\hat{\bm{\phi}}(\bm{s}_y);\bm{s}_y)}_{C}.
\end{align*}

\noindent\textbf{Term $B$.}
Using $\pi(\bm{\theta}\mid\bm{s})\,p_{\text{train}}(\bm{s})
= p_{\text{train}}(\bm{\theta},\bm{s})$ and the definition
of $\bm{\phi}^{\star}_w$, we have
\begin{align*}
\int \mathcal{L}(\bm{\phi}^{\star}_w;\bm{s})\,w_y(\bm{s})\,
    p_{\text{train}}(\bm{s})\,\mathrm{d}\bm{s}
&= \mathbb{E}_{(\bm{\theta},\bm{s})\sim p_{\text{train}}}
    \big[-w_y(\bm{s})\log q_{\bm{\phi}^{\star}_w}
    (\bm{\theta}\mid\bm{s})\big] \\
&= \inf_{\bm{\phi}\in\Phi}\;
    \mathbb{E}_{(\bm{\theta},\bm{s})\sim p_{\text{train}}}
    \big[-w_y(\bm{s})\log q_{\bm{\phi}}
    (\bm{\theta}\mid\bm{s})\big].
\end{align*}
Thus, the first term of $B$ is the infimum of the weighted
objective, and must be (weakly) smaller than the same objective evaluated
at $\hat{\bm{\phi}}(\bm{s}_y)$. Hence $B \le 0$, and this term can be dropped from the upper bound.

\noindent\textbf{Term $A$.}
Fix $\bm{\phi}\in\Phi$. Expand the difference and separate the variation in $q_{\bm{\phi}}(\bm{\theta}\mid\bm{s})$ from the variation in $\pi(\bm{\theta}\mid\bm{s})$, to obtain
\begin{align*}
&\mathcal{L}(\bm{\phi};\bm{s}_y) - \int \mathcal{L}(\bm{\phi};\bm{s})\,w_y(\bm{s})\,p_{\text{train}}(\bm{s})\,\mathrm{d}\bm{s} \\
&= -\int\big[\log q_{\bm{\phi}}(\bm{\theta}\mid\bm{s}_y)\,\pi(\bm{\theta}\mid\bm{s}_y) - \log q_{\bm{\phi}}(\bm{\theta}\mid\bm{s})\,\pi(\bm{\theta}\mid\bm{s})\big]\,w_y(\bm{s})\,p_{\text{train}}(\bm{s})\,\mathrm{d}\bm{\theta}\,\mathrm{d}\bm{s} \\
&= \underbrace{-\int\big[\log q_{\bm{\phi}}(\bm{\theta}\mid\bm{s}_y) - \log q_{\bm{\phi}}(\bm{\theta}\mid\bm{s})\big]\,\pi(\bm{\theta}\mid\bm{s}_y)\,w_y(\bm{s})\,p_{\text{train}}(\bm{s})\,\mathrm{d}\bm{\theta}\,\mathrm{d}\bm{s}}_{A_1} \\
&\quad \underbrace{-\int\log q_{\bm{\phi}}(\bm{\theta}\mid\bm{s})\,\big[\pi(\bm{\theta}\mid\bm{s}_y) - \pi(\bm{\theta}\mid\bm{s})\big]\,w_y(\bm{s})\,p_{\text{train}}(\bm{s})\,\mathrm{d}\bm{s}\,\mathrm{d}\bm{\theta}}_{A_2}.
\end{align*}

For $A_1$, Assumption~\ref{ass:lipschitz}(i) yields
\[
A_1 \le \bar{C}_1\,\mathbb{E}_{p_{\text{train}}(\bm{s})}
    \big[w_y(\bm{s})\|\bm{s}-\bm{s}_y\|\big].
\]

For $A_2$, add and subtract $\log q_{\bm{\phi}}(\bm{\theta}\mid\bm{s}_y)$ to decompose:
\begin{align*}
A_2 &= \underbrace{-\int\log q_{\bm{\phi}}(\bm{\theta}\mid\bm{s}_y)\,\big[\pi(\bm{\theta}\mid\bm{s}_y) - \pi(\bm{\theta}\mid\bm{s})\big]\,w_y(\bm{s})\,p_{\text{train}}(\bm{s})\,\mathrm{d}\bm{s}\,\mathrm{d}\bm{\theta}}_{A_{2.1}} \\
&\quad \underbrace{-\int\big[\log q_{\bm{\phi}}(\bm{\theta}\mid\bm{s}) - \log q_{\bm{\phi}}(\bm{\theta}\mid\bm{s}_y)\big]\,\pi(\bm{\theta}\mid\bm{s}_y)\,w_y(\bm{s})\,p_{\text{train}}(\bm{s})\,\mathrm{d}\bm{s}\,\mathrm{d}\bm{\theta}}_{A_{2.2}} \\
&\quad + \underbrace{\int\big[\log q_{\bm{\phi}}(\bm{\theta}\mid\bm{s}) - \log q_{\bm{\phi}}(\bm{\theta}\mid\bm{s}_y)\big]\,\pi(\bm{\theta}\mid\bm{s})\,w_y(\bm{s})\,p_{\text{train}}(\bm{s})\,\mathrm{d}\bm{s}\,\mathrm{d}\bm{\theta}}_{A_{2.3}}.
\end{align*}

Applying Assumption~\ref{ass:lipschitz}(i) to $A_{2.2}$,
\[
A_{2.2} \le \bar{C}_1\,\mathbb{E}_{p_{\text{train}}(\bm{s})}
    \big[w_y(\bm{s})\|\bm{s}-\bm{s}_y\|\big].
\]

For $A_{2.3}$, apply Assumption~\ref{ass:lipschitz}(i) and the Cauchy--Schwarz inequality:
\begin{align*}
A_{2.3} &\le \int C_1(\bm{\theta},\bm{\phi})\,\|\bm{s}-\bm{s}_y\|\,\pi(\bm{\theta}\mid\bm{s})\,w_y(\bm{s})\,p_{\text{train}}(\bm{s})\,\mathrm{d}\bm{\theta}\,\mathrm{d}\bm{s} \\
&\le \sqrt{\int C_1^2(\bm{\theta},\bm{\phi})\,w_y(\bm{s})\,p_{\text{train}}(\bm{\theta},\bm{s})\,\mathrm{d}\bm{\theta}\,\mathrm{d}\bm{s}}\;\sqrt{\mathbb{E}_{p_{\text{train}}(\bm{s})}
    \big[w_y(\bm{s})\|\bm{s}-\bm{s}_y\|^2\big]} \\
&= \bar{C}_2\sqrt{\mathbb{E}_{p_{\text{train}}(\bm{s})}
    \big[w_y(\bm{s})\|\bm{s}-\bm{s}_y\|^2\big]},
\end{align*}
where the second inequality uses Cauchy--Schwarz with $\pi(\bm{\theta}\mid\bm{s})\,p_{\text{train}}(\bm{s}) = p_{\text{train}}(\bm{\theta},\bm{s})$, and the equality holds by the definition of $\bar{C}_2$.

For $A_{2.1}$, apply Assumption~\ref{ass:lipschitz}(ii):
\begin{align*}
A_{2.1}\le& \int \int |\log q_\phi(\tt\mid\ss_y)||\pi(\tt\mid\ss_y)-\pi(\tt\mid\ss)|w(\ss)p_{\text{train}}(\ss)\dt\tt\dt\ss\\\le& \int \int |\log q_\phi(\tt\mid\ss_y)|f(\tt\mid\ss_y)\|\ss-\ss_y\|^\kappa w(\ss)p_{\text{train}}(\ss)\dt\tt\dt\ss 
\\\le& \bar{C}_3\mathbb{E}_{\ss\sim p_{\text{train}}(\ss)}[w(\ss)\|\ss-\ss_y\|^\kappa].
\end{align*}
Collecting $A_1$, $A_{2.1}$, $A_{2.2}$, and $A_{2.3}$:
\[
A \le 2\bar{C}_1\,\mathbb{E}_{p_{\text{train}}(\bm{s})}
    \big[w_y(\bm{s})\|\bm{s}-\bm{s}_y\|\big]
+ \bar{C}_2\sqrt{\mathbb{E}_{p_{\text{train}}(\bm{s})}
    \big[w_y(\bm{s})\|\bm{s}-\bm{s}_y\|^2\big]}
+ \bar{C}_3\,\mathbb{E}_{p_{\text{train}}(\bm{s})}
    \big[w_y(\bm{s})\|\bm{s}-\bm{s}_y\|^{\kappa}\big].
\]

\noindent\textbf{Term $C$.}
For any fixed $\bm{\phi}$, the bound derived for Term $A$ applies to $|h(\bm{\phi})|$ where $h(\bm{\phi}) := \mathcal{L}(\bm{\phi};\bm{s}_y) - \int \mathcal{L}(\bm{\phi};\bm{s})\,w_y(\bm{s})\,p_{\text{train}}(\bm{s})\,\mathrm{d}\bm{s}$.
Since $C = -h(\hat{\bm{\phi}}(\bm{s}_y))$, we have the upper bound:
\[
C \le 2\bar{C}_1\,\mathbb{E}_{p_{\text{train}}(\bm{s})}
    \big[w_y(\bm{s})\|\bm{s}-\bm{s}_y\|\big]
+ \bar{C}_2\sqrt{\mathbb{E}_{p_{\text{train}}(\bm{s})}
    \big[w_y(\bm{s})\|\bm{s}-\bm{s}_y\|^2\big]}
+ \bar{C}_3\,\mathbb{E}_{p_{\text{train}}(\bm{s})}
    \big[w_y(\bm{s})\|\bm{s}-\bm{s}_y\|^{\kappa}\big].
\]

Adding the bounds for $A$, $B$ (where $B \le 0$), and $C$ yields the result.
\end{proof}

The following is a direct consequence.

\begin{corollary}\label{cor:indicator_gap}
If $w_y(\bm{s}) \le C'\,\mathbb{I}\{\|\bm{s}-\bm{s}_y\|\le\epsilon\}$ for some $C'>0$ and $\epsilon>0$, then under Assumption~\ref{ass:lipschitz},
\[
\Delta_{\mathrm{am}}(\bm{s}_y) \le C\big[\epsilon \vee \epsilon^{\kappa}\big],
\]
where $C>0$ depends on $C'$, $\bar{C}_1$, $\bar{C}_2$, $\bar{C}_3$, and $\kappa$.
\end{corollary}


\end{document}